\documentclass[a4paper,11pt,twoside]{article}
\usepackage{amsmath}
\usepackage{amsthm}
\usepackage{amssymb}
\usepackage{amsfonts}
\usepackage{fullpage}
\usepackage{graphicx,color}
\usepackage[dvipsnames]{xcolor}
\usepackage{fancybox}
\usepackage{tikz}
\usepackage{subcaption}
\usepackage{hyperref}
\usetikzlibrary{positioning,chains,fit,shapes,calc}
\usetikzlibrary{trees}
\usetikzlibrary{decorations.pathmorphing}
\usetikzlibrary{decorations.markings}

\usepackage{cool}
\Style{DSymb={\mathrm d},DShorten=true,IntegrateDifferentialDSymb=\mathrm{d}}
\usepackage{todonotes}

\newtheorem{theorem}{Theorem}

\newtheorem{lemma}[theorem]{Lemma}
\newtheorem{proposition}[theorem]{Proposition}
\newtheorem{claim}[theorem]{Claim}
\newtheorem{definition}[theorem]{Definition}

\newcommand{\abs}[1]{\left\vert#1\right\vert}
\newcommand{\set}[1]{\left\{#1\right\}}
\newcommand{\eps}{\varepsilon}

\newcommand{\CommentS}[1]{}
\newcommand{\CSP}{\#{\rm CSP}}

\title{FPTAS for Counting Monotone CNF  \footnote{This work was performed when the first author was an intern at Microsoft Research Asia}}

\author{ Jingcheng Liu
        \thanks{Shanghai Jiao Tong University. {\tt liuexp@gmail.com}}
        \and
        Pinyan Lu\thanks{Microsoft Research. {\tt pinyanl@microsoft.com}}
}
\date{}
\begin{document}

\maketitle

\begin{abstract}
A monotone CNF formula is a Boolean formula in conjunctive normal form where each variable appears positively.
We design a deterministic fully polynomial-time approximation scheme (FPTAS) for counting the number of satisfying assignments for a given monotone CNF formula when each variable appears in at most $5$ clauses.
Equivalently, this is also an FPTAS for counting set covers where each set contains at most $5$ elements.
If we allow variables to appear in a maximum of $6$ clauses (or sets to contain $6$ elements), it is NP-hard to approximate it.
Thus, this gives a complete understanding of the approximability of counting for monotone CNF formulas.
It is also an important step towards a complete characterization of the approximability for all bounded degree Boolean \#CSP problems.
In addition, we study the hypergraph matching problem, which arises naturally towards a complete classification of bounded degree Boolean \#CSP problems,
and show an FPTAS for counting 3D matchings of hypergraphs with maximum degree $4$ .

Our main technique is correlation decay, a powerful tool to design deterministic FPTAS for counting problems defined by local constraints among a number of variables.
All previous uses of this design technique fall into two categories: each constraint involves at most two variables, such as independent set, coloring, and spin systems in general; or each variable appears in at most two constraints, such as matching, edge cover, and holant problem in general.
The CNF problems studied here have more complicated structures than these problems and require new design and proof techniques.
As it turns out, the technique we developed for the CNF problem also works for the hypergraph matching problem.
We believe that it may also find applications in other CSP or more general counting problems.

\end{abstract}

\setcounter{page}{0}\thispagestyle{empty}
\newpage

\section{Introduction}
We study the complexity of approximately counting the number of satisfying assignments of a given Boolean formula.
For any given parameter $\epsilon>0$, the algorithm outputs a number $\hat{N}$ such that $(1-\epsilon) N\leq \hat{N} \leq (1+\epsilon) N$, where $N$ is the accurate number of solutions for the given formula.
We also require that the running time of the algorithm be bounded by $poly(n,1/\epsilon)$, where $n$ is the size of the formula.
This is called a fully polynomial-time approximation scheme (FPTAS). The randomized relaxation of FPTAS is called fully polynomial-time randomized approximation scheme (FPRAS), which uses random bits in the algorithm and requires that the final output be within the range $[(1-\epsilon) N, (1+\epsilon) N]$ with high probability.
Many interesting combinatorial problems can be described by Boolean formulas. However, for many of them such as SAT, {\sc Exact-One}, {\sc Not-All-Equal} among others, it is already NP-hard to determine whether a satisfying assignment exists. For these problems, we cannot get a polynomial time algorithm to count or approximately count (since we cannot relatively  approximate zero) the number of solutions unless NP=P. Therefore, we mainly focus on those problems for which there is always a satisfying assignment or we have a polynomial time algorithm to determine that. One famous such example is formulas in disjunctive normal form (DNF). It is easy to determine if a DNF formula is satisfiable or not. Basically, a DNF formula is always satisfiable except in trivial cases where each clause contains a contradiction ($x$ and $\bar{x}$). There is an FPRAS for counting the number of satisfying assignments of any given DNF formula~\cite{KarpL83,KarpLM89}. It is an important open question to derandomize the algorithm~\cite{Trevisan04,gopalan2012dnf}.

Conjunctive normal form (CNF) is more widely applicable than DNF. But the decision version is already NP-hard.
One interesting sub-family of CNF is monotone CNF where each variable appears positively in clauses.
For monotone CNF, the decision version is trivial: we can simply set all variables as True to satisfy the formula.
Therefore, it is an interesting problem to count the number of solutions.
Moreover, monotone CNF is already quite expressive as any monotone Boolean function can be expressed as a monotone CNF.
It also contains numerous interesting combinatorial problems as special cases, so long as the combinatorial problem is defined by local constraints and the feasible sets are either downward closed or upward closed, which is typical for many combinatorial problems.
For example, vertex cover (or complementary independent set) problem can be viewed as monotone 2CNF; edge cover can be viewed as read twice monotone CNF where each variable appears no more than  twice.  Indeed, monotone CNF is exactly the same as set cover problem, where variables are sets and clauses are elements.

In a previous work, we design an FPTAS for counting edge covers for any given graphs~\cite{counting-edge-cover}. For counting vertex covers (or independent sets), there is an FPTAS if the maximum degree is 5~\cite{Weitz06} and it is NP-hard even for 6-regular graphs~\cite{Sly10}. As counting independent sets is a special case of counting monotone CNF, we have the following hardness result.

\begin{proposition}
There is no FPRAS (or FPTAS) to count monotone CNF if a variable can appear in $6$ clauses unless NP=RP.
\end{proposition}

Given that, the best hope is to get an FPTAS for monotone CNF formulas where each variable appears at most 5 times. The main result of this paper is indeed such an algorithm.

\begin{theorem}
There is FPTAS to count monotone CNF if each variable appears at most $5$ times.
\end{theorem}

This algorithm can also be interpreted as an FPTAS for counting the number of set covers when each set contains at most five elements.
In particular this covers the aforementioned FPTAS's for counting edge covers and independent sets.
As for these two special cases, the main technique is also correlation decay. First, we use a probability distribution point of view for the counting problem.
Given a CNF formula, we consider a uniform distribution over all its satisfying assignments,
which induces a marginal probability for a variable's assignment.
There is a standard procedure to compute the number of solutions from these marginal probabilities.
Our task is to estimate these marginal probabilities.
We establish a computation tree to relate the marginal probability of a variable to that of its neighbors.
For counting independent sets, edge covers and all other problems (to the best of our knowledge) for which an FPTAS was designed using correlation decay technique, there is only one layer of neighbors.
More concretely, all these problems belong to one of two families:
either each constraint involves at most two variables such as independent set~\cite{Weitz06}, coloring~\cite{GK07,LY13},
and spin systems in general~\cite{LLY12,SST,LLY13,GK07,LY13}; or each variable appears in at most two constraints
such as matching~\cite{BGKNT07}, edge cover~\cite{counting-edge-cover}, and holant problems~\cite{YZ13,fibo-approx} in general.
In CNF formulas, each variable appears in multiple clauses and each clause involves multiple variables.
To handle this, one recursion step in our computation tree has a two-layer structure. In the first layer, we deal with different occurrences of a variable using a similar idea for the self-avoiding walk tree in~\cite{Weitz06}. In the second layer, we deal with individual variables in each clause using a similar computation tree as that in~\cite{counting-edge-cover}. Then, we prove a correlation decay property with respect to this computation tree, which means that the nodes in the computation tree that are far from the root have little influence on the marginal probability of the root.
Based on this property, we can truncate the computation tree and get a good estimation of the marginal probability in polynomial time.
The arising of a two-layer group structure significantly complicates the analysis.
We introduce a sub-additivity argument to deal with variables in different groups separately, and carry out different treatments for groups with different sizes.
In particular, for groups with a super constant size, we have to truncate them early in order to keep the total size of the computation tree within polynomial. To do that, we employ a stronger notion called computationally efficient correlation decay, which is introduced in~\cite{LLY12} and also successfully used in~\cite{counting-edge-cover}.

Based on the same idea of alternating with a two-layer recursion, we also provide FPTAS for an additional counting CSP problem, which also implies an FPTAS for 3D matching.
\begin{theorem}
	There is an FPTAS for counting 3D matchings of hypergraph with maximum degree at most $4$.
\end{theorem}

As an aside, the problem of counting hypergraph matching could be tranformed into counting independent sets over the line graph of the hypergraph,
where the vertices of the line graph are the hyperedges, and two vertices has an edge if those two hyperedges intersect.
However, the two problems are not equivalent when the maximum degree comes into play.
For instance, 3D matching with maximum degree $4$ translates to counting independent sets with maximum degree $9$, for which no FPTAS in general is possible unless $P=NP$.
Still, by leveraging the locally clique-like structure of line graphs with a two-layer recursion, we have been able to show an FPTAS.
In particular, instead of alternating between clauses and variables as in the CNF problem, the two-layer recursion here will alternate between hyperedges and vertices. In addition, what we will show under the Boolean Constraint Satisfaction Problem (CSP) framework is actually a slightly stronger result, where non-uniform hypergraph is also allowed.

Hypergraph matching (or set packing) can be viewed as a dual version of monotone CNF (or set cover).
 In fact, all these problems are just special case of monotone Boolean CSP. 
 Specifically, monotone CNF is a CSP with the Boolean OR function as constraint, and hypergraph matching is a CSP with the {\sc At-Most-One} constraint.
 Not only are they interesting counting problems on their own, they also play important roles in the classification of approximability for the bounded degree Boolean \#CSP~\cite{DyerGJR12}.
We use $\CSP_d(\Gamma)$ to denote the problem of counting the number of solutions for a Boolean CSP where all the constraints are from $\Gamma$ and each variable appears at most $d$ times.
Then $\CSP_d(OR)$ is exactly our Read-$d$-Mon-CNF.
As mentioned above, there is an FPTAS if $d\leq 5$ and the problem is NP-hard if $d\geq 6$,
so a broader view would be to ask the same question for other $\Gamma$ and $d$. If $d=\infty$, i.e. no degree bound, then \cite{DyerGJ10} gave a complete classification in terms of $\Gamma$: the problem is NP-hard under randomized reduction, BIS-hard (as hard as approximately counting independent sets for a bipartite graph), or polynomial time computable  even for exact counting.
Basically, if we believe these hardness assumptions, there is no interesting approximable cases.
The bounded degree case was studied in \cite{DyerGJR12}, with a slightly technical assumption: in addition to the constraints in $\Gamma$, two unary constant constraints (pinning) are always available. We write $\CSP^c_d(\Gamma)$ for $\CSP_d(\Gamma)$ with this assumption.
They then studied the approximability of $\CSP^c_d(\Gamma)$ in terms of $\Gamma$ and $d$.
For $d\ge 6$, a similar classification as that of unbounded degree case was obtained.
For $d\in \{3,4,5\}$, a partial classification was also given in~\cite{DyerGJR12}, in which the monotone Boolean constraint stands out as the only unknown case.
It is also the only interesting family in the sense that all other problems in the framework $\CSP^c_d(\Gamma)$ with $d\ge 3$ are either hard (NP-hard or BIS-hard) to approximately count or polynomial time computable even for exact counting.
Thus, our above FPTAS for Read-$5$-Mon-CNF falls in this interesting range, and makes an important step towards a full classification for $\CSP^c_d(\Gamma)$.
In section \ref{sec:csp}, we discuss the implication of our FPTAS in the classification and also obtain new hardness result and FPTAS for other monotone Boolean $\CSP$ problems. 
It is worth noting that our two-layer recursion, with the newly developed analysis technique for the CNF problem, are also used to design and prove the additional FPTAS for a general $\CSP$ problem. This structure is indeed common for general $\CSP$ problems and the new techniques developed here may find applications in other problems.

 \subsection*{Related Work}
The approach to designing FPTAS via correlation decay is introduced in \cite{BG08} and \cite{Weitz06}. The most successful example is for anti-ferromagnetic two-spin systems~\cite{LLY12,SST,LLY13}, including counting independent sets~\cite{Weitz06}. The correlation decay based FPTAS is beyond the best known MCMC based FPRAS and achieves the boundary of approximability~\cite{SS12,galanis2012inapproximability}. The approach was also extended to count colorings and compute the partition function of multi-spin system~\cite{GK07,LY13}.

There is also a beautiful long line of research  on designing FPRAS for approximate counting by sampling and most successfully sampling by Markov chain, see for example~\cite{MC_JA96,JS93, app_JSV04, app_GJ11, app_DJV01, col_Jerrum95, col_Vigoda99, IS_DFJ02, IS_DG00, IS_LV97}.

For counting matchings, an FPRAS based on Markov Chain Monte Carlo (MCMC) is known for any graph\cite{jerrum1989approximating}, and deterministic FPTAS is only known for graphs with bounded degree\cite{BGKNT07}.
Hypergraph matching is also known as set packing.
More recently, an independent result for approximately counting 3D matching was also obtained in \cite{33matching} for maximum degree $3$.

\section{Preliminary}

\begin{definition} [Read-$d$-Mon-CNF]
  A  read d times monotone CNF formula (Read-$d$-Mon-CNF) is a CNF formula where every literal is positive occurrence of some variable and each variable appears in at most $d$ clauses.
Formally we write a monotone CNF formula as $C = \bigwedge_j c_j$ where $c_j = \bigvee_i x_{j,i}$, and $x_{j,i}$ are (not necessarily distinct) variables.

\end{definition}

A satisfying assignment for a CNF formula is an assignment to the variables (True or False) such that all clauses are satisfied.
Counting the number of such satisfying assignments is our main concern.
We will use numeric value $1$ to indicate Boolean value \emph{True}, and $0$ for \emph{False}.

We denote the occurrences of variable $x$ in formula $C$ by $d_x(C)$, then $C$ being Read-$d$ is the same as $\forall x, d_x(C) \leq d$.
Let $\abs{c}$ be the number of distinct variables in a clause $c$, a \emph{singleton clause} is a clause $c$ with $\abs{c} = 1$.
A monotone CNF formula is \emph{well-formed} if each clause does not contain duplicate variables, there is no singleton clause, and no clause is a subset of another. Any Read-$d$-Mon-CNF formula can be re-written as a well-formed Read-$d$-Mon-CNF formula.


A general Boolean constraint with arity $k$ is a mapping $C: \set{0,1}^k \to \set{0,1}$. Two special unary constraints are called pinning:
$\Delta_0,\Delta_1$ defined by $\Delta_0(0)=1, \Delta_0(1)=0, \Delta_1(0)=0 $ and  $\Delta_1(1)=1$.
Basically, the constraints $\Delta_0(x)$ and $\Delta_1(x)$ fix the variable $x$ to be $0$ and $1$ respectively, which we call the variable is pinned to $0$ and $1$ respectively.
A Boolean constraint $C$ is monotone if $\forall \mathbf{x}, \mathbf{y} \in \set{0,1}^k$, $\mathbf{x} \leq \mathbf{y} \Rightarrow C(\mathbf{x}) \leq C(\mathbf{y})$.
The other direction of monotone is equivalent to this after switching the name of $0$ and $1$. Thus, we focus on this direction in the paper and our conclusion also holds for the other direction by simply renaming $0$ and $1$.
Except for those trivial constant functions and some pinned variables, a monotone Boolean constraint can always be re-written as a unique well-formed monotone CNF~\cite{DyerGJR12}.

Let $\Gamma$ be a set of Boolean constraints and $d>0$ be an integer, we use $\CSP_d(\Gamma)$ to denote the problem of counting the number of solutions for a Boolean CSP problem where each constraint is from the set $\Gamma$ and each variable appears in at most  $d$ constraints.
We write $\CSP^c_d(\Gamma)\triangleq \CSP_d(\Gamma \cup \{\Delta_0,\Delta_1\})$, where
 $\Delta_0$ and $\Delta_1$ are always assumed to be available. 

%
%
%
%
%

Recall the equivalence between monotone CNF formulas and set covers, we will also treat a CNF formula as a set of clauses, and a clause as a set of variables,
and define some set operations for CNF $C=\bigwedge_j c_j $, clause $c=\bigvee_i x_{i}$, and variable $x$ as follows:
\begin{itemize}
	\item $c - x \triangleq \bigvee_{i: x_i \neq x} x_i$, the occurence of $x$ in $c$ (if any) is pinned to $0$;
	\item $C - c \triangleq \bigwedge_{j: c_j \neq c} c_j$, remove the clause $c$ from $C$;
	\item $C + c \triangleq \left(\bigwedge_j c_j \right) \wedge c$ , add a clause $c$ to $C$;
	\item $C - x \triangleq \bigwedge_j \left( c_j - x \right)$ , the variable $x$ is pinned to $0$ in the entire formula $C$.
\end{itemize}

We will also write $C - x - y \triangleq (C-x) -y$.
In general we use $n$ to refer the number of variables, and $m$ for the number of clauses (or constraints).
We use $\mathbf{1}$ for the all-one vector, and a $d$-dimensional vector $\mathbf{t}$ is also written as $\set{t_i}_{i=1}^d$, with the $i$-th coordinate being $t_i$, so $\set{t_i} = \mathbf{0}$ means $\forall i, t_i=0$. We write $\underline{w}= \set{w_i}$ for an ascending ordered sequence.
We also write $\mathbf{t} \setminus \set{t_d} \triangleq \set{t_1,t_2,\ldots,t_{d-1}}$ as a $(d-1)$-dimensional vector after projection. 

\section{The Algorithm}
Given a monotone CNF formula $C$, let $X(C)$ be the set of all satisfying assignments for $C$.
We associate a uniform distribution on $X(C)$, which induces a marginal probability $\mathbb{P}_C ( x=0)$: the probability that $x$ is assigned to be $0$ (False) if we take a satisfying assignment from $X(C)$ uniformly at random.
We use $R(C,x)$ to denote its ratio: $R(C,x) \triangleq \frac{\mathbb{P}_C ( x=0)}{\mathbb{P}_C ( x=1)}$.
In this section, we shall give an algorithm to compute these marginal probabilities and count the size of $X(C)$.

\subsection{Recursion}
First we prove a recursive relation which relates $R(C,x)$ to that of smaller instances.

\begin{lemma}
\label{recursion1}
Let $d \triangleq d_x(C)$, and the $d$ clauses containing $x$ be enumerated as $\set{c_j}_{j=1}^d$.
Denote $w_j \triangleq \abs{c_j} - 1$,  $C_j \triangleq \left(C - \sum_{k\neq j} c_k \right) + \sum_{k=j+1}^d (c_k - x)$.
Let $\set{x_{j,i}}_{i=1}^{w_j}$ be the set of variables in $(c_j - x)$,
 and $C_{j,i} \triangleq C_j - c_j -\sum_{k=1}^{i-1} x_{j,k} $. Then we have
\begin{equation}\label{equ:rec}
R(C,x) = \prod_{j=1}^{d} \left( 1 - \prod_{i=1}^{w_j} \frac{R(C_{j,i},x_{j,i})}{1+R(C_{j,i},x_{j,i})} \right).
\end{equation}
\end{lemma}

As an aside, $C_j$ is obtained from $C$ by pinning the occurrences of $x$ in $c_1, c_2, \cdots, c_{j-1}$ to $1$ (and these clauses are thus removed since they have been satisfied) and the occurrences of $x$ in $c_{j+1}, c_{j+2}, \cdots, c_{d}$ to $0$ (and thus we simply remove variable $x$ from these clauses).
$C_{j,i}$ is further obtained from $C_j$ by removing clause $c_j$ and pinning $x_{j,1}, x_{j,2}, \cdots, x_{j,i-1}$ (all occurrences) to $0$.

Alternatively viewing $C$ as a set cover instance, $x$ is a set containing $\set{c_k}_k^d$ as elements, $C_j$ is obtained from $C$ by first removing elements $c_1, c_2, \cdots, c_{j-1}$ entirely, and then removing elements $c_{j+1}, c_{j+2}, \cdots, c_{d}$ only from the set $x$. Then $C_{j,i}$ is to remove element $c_j$, and then the sets $x_{j,1},x_{j,2}, \cdots, x_{j,i-1}$ from $C_j$.

\begin{proof}
If $d=0$, $x$ is entirely a free variable, so $R(C,x) = 1$, and recall that we adopt the convention that if $d=0$, the product is also $1$. If $w_j=0$ for some $j\leq d$, the clause $c_j$ is a singleton with variable $x$. As a result, $x=1$ and thus $R(C,x) = 0$. In this case, the identity is also true since the $j$-th internal product is $1$ for $w_j=0$ and we get $R(C,x) = 0$.

In the following, we assume $d \ge 1$ and $w_j\geq 1$ for $j=1, 2, \cdots, d$. Note that in this case all the new instances $C_{j,i}$ are well-defined monotone CNF formulas. We substitute the $d$ occurrences of $x$ with $d$ independent new variables $\set{\tilde{x}_j}_{j=1}^d$ (the occurrence of $x$ in $c_j$ is replaced by $\tilde{x}_j$), and denote this new CNF formula by $C'$.  We have
\[
R(C,x)=\frac{\mathbb{P}_C ( x=0)}{\mathbb{P}_C ( x=1)}
= \frac{\mathbb{P}_{C'} \left( \set{\tilde{x}_j} = \mathbf{0} \right)}{\mathbb{P}_{C'} \left( \set{\tilde{x}_j} = \mathbf{1} \right)}
= \prod_{j=1}^d \frac{\mathbb{P}_{C'} \left( \set{\tilde{x}_i}_{i=1}^{j-1} = \mathbf{1},  \set{\tilde{x}_i}_{i=j}^{d} = \mathbf{0} \right)}{\mathbb{P}_{C'} \left(  \set{\tilde{x}_i}_{i=1}^{j} = \mathbf{1},  \set{\tilde{x}_i}_{i=j+1}^{d} = \mathbf{0}  \right)}
= \prod_{j=1}^d R(C_j, \tilde{x}_j). \]
%

Now we further expand $R(C_j, \tilde{x}_j)$.  Since $\tilde{x}_j$ is a newly introduced variable, it only appears once in $C_j$ which is in $c_j$. As $w_j = \abs{c_j} - 1$,  $\set{x_{j,i}}_{i=1}^{w_j}$ is the set of variables in $c_j - \tilde{x}_j$,
 using the fact that $c_j$ is a monotone clause and $\tilde{x}_j$ does not appear in other clauses,  we have
\[ R(C_j, \tilde{x}_j) =
\frac{\mathbb{P}_{C_j} ( \tilde{x}_j =0)}{\mathbb{P}_{C_j} ( \tilde{x}_j =1)} =
1 - \mathbb{P}_{C_j-c_j} \left( \set{x_{j,i}}_{i=1}^{w_j} = \mathbf{0} \right)=
1 - \prod_{i=1}^{w_j} \mathbb{P}_{C_{j,i}} \left( x_{j,i} = 0 \right),\]
By substituting $\mathbb{P}_{C_{j,i}} \left( x_{j,i} = 0 \right) = \frac{R(C_{j,i},x_{j,i})}{1+R(C_{j,i},x_{j,i})}$, this concludes the proof.
\end{proof}

For all new instances $R(C_{j,i},x_{j,i})$ involved in the recursion, one occurrence of $x_{j,i}$ in the original formula $C$ (which is in $c_j$) is eliminated in $C_{j,i}$ as we have removed $c_j$ from $C$ to get $C_{j,i}$. Therefore,
 for Read-$5$-Mon-CNF $C$, $d_{x_{j,i}} ( C_{j,i} ) \leq 4$. In other words, the recursion with $d=5$ is invoked no more than once,
 which would be at the initial step.


\begin{proposition}
$R(C,x) \leq 1$.
\end{proposition}
\begin{proof}
 Since $C$ is a monotone CNF, a satisfying assignment with $x=0$ can be injectively mapped to a satisfying assignment with $x=1$,
	thus we have $\mathbb{P}(x=0) \leq \mathbb{P}(x=1)$.
\end{proof}

\subsection{Truncated Computation Tree}
In the recursion (\ref{equ:rec}) of Lemma \ref{recursion1}, we can recursively expand these $R(C_{j,i},x_{j,i})$ to those of smaller
instances until we reach trivial instances, which yields a computation tree to compute $R(C,x)$. However, the total size of the computation tree can be exponential.
%
Here we estimate by truncating it up to recursion depth $L$.
 Formally, for a Read-$5$-Mon-CNF $C$, a variable $x$  with $d \triangleq d_x(C)$ and a non-negative integer $L$,
 we recursively define and compute $R(C,x,L)$ as: 
\[R(C,x,L) =
\left\{
  \begin{array}{ll}
    0, & \hbox{\textrm{$w_j = 0$ for some $j$};} \\
      1, & \hbox{$d=0$ or $L = 0$;} \\
    \prod_{j=1}^{d} \left( 1 - \prod_{i=1}^{w_j}  \frac{R\left(C_{j,i},\ x_{j,i},\ \max\left(0,L - \lceil\log_4 (w_j + 1) \rceil\right)\right)}{1+R\left(C_{j,i},\ x_{j,i},\ \max\left(0, L - \lceil\log_4 (w_j + 1) \rceil\right)\right)} \right), & \hbox{$1 \leq d \leq 5$.}
  \end{array}
\right.
\]

The recursion depth $L$ used here is known as $M$-based depth introduced in \cite{LLY12} with $M=4$. 
As for the computation tree, another view of the recursion depth $L$ is we replace every node with a branching degree greater than $4$ with a $4$-ary branching subtree.
Now it is easy to see that the nodes involved in the branching computation tree up to depth $L$ are at most $O((4d)^L) = O(16^L)$,
and for second-to-base-case nodes (i.e. nodes with $0<L \leq \lceil \log_4{(d+1)} \rceil $ ) they involve at most $O(n)$ extra base cases,
so the running time for the algorithm to compute $R(C,x,L)$ is $O(n 16^L)$

Recall that the $d=5$ case is invoked only once, then the algorithm keeps exploring the recursion with $1\leq d\leq 4$ until it hits one of the three boundary cases. We remark that the two boundary cases $d=0$ and  $w_j = 0$ for some $j$ can be covered by the recursion automatically as we define an empty product to be $1$. The values for these boundary cases are indeed accurate
(equal to $R(C,x)$)
and we list them out separately to be more explicit.  For another boundary case $L=0$, we choose the value $1$ here which is an arbitrary guess.
Indeed any number between $0$ and $1$ works, as we shall prove the correlation decay property,
which states that the value for $R(C,x,L)$ with large $L$ is almost independent with the choice of these values for $L=0$.
Formally,
%
%
%
we have the following key lemma, for which the proof is laid out in Section \ref{amortized-proof}.
\begin{lemma}[Correlation decay]
\label{main-lem}
Let $\alpha = 0.981$,  $C$ be a Read-$5$-Mon-CNF formula   and $x$ be a variable of $C$. Then
\begin{equation}
\label{main-lem-equation}
\abs{ R(C,x,L) -  R(C,x)} \leq 5 \sqrt{6} \alpha^L.
\end{equation}
\end{lemma}

\subsection{Counting Algorithm}
With these truncated marginal probability ratios $R(C,x,L)$, we derive our counting algorithm.
Let $Z(C)$ be the number of satisfying assignments of $C$,
and $\set{x_i}_{i=1}^n$ be an enumeration of variables in any order in a Read-$5$-Mon-CNF $C$.
As a monotone CNF, $\set{x_i}_{i=1}^n=\mathbf{1}$ is a satisfying assignment. Now with $\set{x_i}_{i=1}^n$ sampled uniformly,  $\mathbb{P} (\set{x_i}_{i=1}^n = \mathbf{1}) $ has two expressions:
	\[ \mathbb{P} ( \set{x_i}_{i=1}^n = \mathbf{1} ) = \frac{1}{Z(C)} \textrm{\ \ and\ \ } \mathbb{P} ( \set{x_i}_{i=1}^n = \mathbf{1} ) = \prod_i \mathbb{P}_C \left( x_i = 1 \mid \set{x_j}_{j=1}^{i-1} = \mathbf{1} \right) =  \prod_i \mathbb{P}_{C_i} \left( x_i = 1 \right),\]
where $C_1=C$ and $C_{i+1}$ is obtained from $C_{i}$ by pinning $x_i$ to $1$ for $i=1,2,\cdots, n-1$.
	By substitution and $\mathbb{P}_{C_i}\left( x_i = 1 \right) = \frac{1}{1 + R(C_i, x_i) }$, we get
\[ Z(C) = \prod_i \left( 1 + R(C_i, x_i) \right).  \]
Hence we also get
$Z(C, L) \triangleq \prod_i \left( 1 + R(C_i, x_i, L) \right)$
as an estimation for $Z(C)$.
By correlation decay lemma, we show that
\begin{theorem}
\label{counting-thm1}
Let $0< \eps < 1$, $C$ be a Read-$5$-Mon-CNF with $n$ variables, $L = \log_{\alpha} \left(\frac{\eps}{10 \sqrt{6} n} \right)$. Then
$Z(C,L)$ is the desired FPTAS for $Z(C)$ with running time $O\left(n^2 \left(\frac{n}{\eps}\right)^{\log_{1/ \alpha} (16)} \right)$.
\end{theorem}

%

Since it is a common procedure to carry out the proof details, we leave it to the appendix.



\section{Correlation Decay}
\label{amortized-proof}
In this section, we shall prove  Lemma \ref{main-lem}, the key correlation decay lemma.
To prove such an exponential correlation decay, the most common method is to use induction and prove that the error is decreased by a constant factor along each recursion step.
Unfortunately, this is not true in our case. Instead we perform an  amortized analysis on the decay rate by a potential function.
We choose $\varphi(x) \triangleq 2 \sinh ^{-1}\left(\sqrt{x}\right)$
to map the values $R(C,x,L), R(C,x)$ into a new domain and prove the following inequality after mapping:
\begin{equation}
\label{main-lem-equation2}
\abs{ \varphi \circ R(C,x,L) -  \varphi \circ R(C,x)} \leq 5 \alpha^L.
\end{equation}
The fact that condition (\ref{main-lem-equation2}) $\implies$ (\ref{main-lem-equation}) is due to Mean Value Theorem.
Since $0\leq R(C,x,L), R(C,x) \leq 1$, we have $0\leq \varphi \circ R(C,x,L),  \varphi \circ R(C,x) < 2$. As a result,  $\exists \bar{y}: 0 \leq \bar{y} < 2$ such that
\[  |R(C,x,L) -   R(C,x)| = \frac{\partial \varphi^{-1}(y)}{\partial y}\Big|_{y=\bar{y}} \cdot  | \varphi \circ R(C,x,L) -  \varphi \circ R(C,x) | \leq \sqrt{6} \cdot 5 \alpha^L= 5 \sqrt{6} \alpha^L, \]
where the inequality uses the fact that  $\frac{\partial \varphi^{-1}(y)}{\partial y} = \sqrt{y(1+y)} < \sqrt{6}$ for $0\leq y < 2$.

Since the case $d_x(C)=5$ is applied at most once at the root, we have $d_x(C)\leq 4$ for all the other nodes. 
Hence we first prove the following stronger bound for the $d_x(C)\leq 4$ case:

\begin{equation}\label{equ:bound-d4}
\abs{ \varphi \circ R(C,x,L) -  \varphi \circ R(C,x)} \leq 2 \alpha^L, \ \mathrm{ for }\ d_x(C) \leq 4.
\end{equation}
The fact that (\ref{equ:bound-d4}) $\implies$ (\ref{main-lem-equation2}) shall be shown later in (\ref{main-lem-d5}).
Now we prove (\ref{equ:bound-d4}) by induction on $L$.
For the base case $L=0$,  since $0\leq \varphi \circ R(C,x,L),  \varphi \circ R(C,x) < 2$, it is clear that  $\abs{ \varphi \circ R(C,x,L) -  \varphi \circ R(C,x)}< 2$.

Supposing the induction hypothesis holds for $L<l$, we prove it is true for $L=l$.
If $x$ is a free variable in $C$, i.e. $d=0$, $R(C, x, L) = R(C, x) = 1$. 
And if $x$ can be inferred (due to $w_j=0$ for some $j$), $R(C,x,L) = R(C,x) = 0$. 
In the following, we assume that $1\leq d \leq 4$ and $w_j\geq 1$.

Denote
$ h(\mathbf{r}) \triangleq \prod_j^d \left( 1 - \prod_i^{w_j} \frac{r_{j,i}}{1+r_{j,i}} \right) $, which is the analytic version of the recursion.
Let $\mathbf{y}$ be the true vector with $y_{j,i} = \varphi \circ R\left(C_{j,i},\ x_{j,i}\right)$ and $\mathbf{\hat{y}}$ be the estimated vector with $\hat{y}_{j,i} = \varphi \circ R\left(C_{j,i},\ x_{j,i},\ \max(0,L - \lceil \log_4 (w_j + 1) \rceil) \right)$.
We abuse notations here and denote $\mathbf{r} \triangleq \varphi^{-1} ( \mathbf{y} )$ for $r_{j,i} = \varphi^{-1}(y_{j,i})$, which is applying $\varphi^{-1}$ entry-wise to $\mathbf{y}$, similarly for $\mathbf{\hat{r}} \triangleq \varphi^{-1} ( \mathbf{\hat{y}} )$.
Then $\varphi \circ R(C,x,L)=\varphi \circ h (\mathbf{r})$ and $\varphi \circ R(C,x)=\varphi \circ h (\mathbf{\hat{r}})$.

Now by Mean Value Theorem, $\exists \gamma: 0\leq \gamma \leq 1, \mathbf{\tilde{y}} =\gamma \mathbf{y} + (1-\gamma) \mathbf{\hat{y}}$ such that, let $\mathbf{\tilde{r}} \triangleq \varphi^{-1} (\mathbf{\tilde{y}}) $,
\begin{align*}
\varphi\circ R(C,x,L) - \varphi\circ R(C,x) & = \sum_{a,b} \frac{ \partial (\varphi \circ h \circ \varphi^{-1}) }{\partial y_{a,b}}\Big|_ {\mathbf{y}=\mathbf{\tilde{y}}} \cdot( \hat{y}_{a,b}- y_{a,b})  \\
&=\sum_{a,b} \frac{\Phi(h(\mathbf{\tilde{r}}))}{\Phi(\tilde{r}_{a,b})} \left( \frac{\partial h}{\partial \tilde{r}_{a,b}}\Big|_{\mathbf{r}=\mathbf{\tilde{r}}}\right) \cdot ( \hat{y}_{a,b}- y_{a,b}),
\end{align*}
where $\Phi(x) \triangleq \D{\varphi(x)}{x}= \frac{1}{\sqrt{x(1+x)}}$. Now by induction hypothesis, we have
\[|\hat{y}_{a,b}- y_{a,b}|\leq 2 \alpha^{\max(0,L - \lceil \log_4 (w_a + 1) \rceil)} \leq 2 \alpha^{L - \lceil \log_4 (w_a + 1) \rceil}.\]
Substituting this into the above equation, we have
\begin{align*}
|\varphi\circ R(C,x,L) - \varphi\circ R(C,x)| \leq
2 \alpha^L \sum_{a,b} \frac{\Phi(h(\mathbf{\tilde{r}}))}{\Phi(\tilde{r}_{a,b})} \Big| \frac{\partial h}{\partial \tilde{r}_{a,b}}\Big|  \alpha^ {- \lceil \log_4 (w_a + 1) \rceil}.
\end{align*}
%
%
Let $\underline{w}=\set{w_a}_{a=1}^d$, it is sufficient to show that the amortized decay rate
\[\hat{\kappa}_d^{\underline{w}} (\mathbf{r}) \triangleq \sum_{a,b} \frac{\Phi(h(\mathbf{r}))}{\Phi(r_{a,b})} \Big| \frac{\partial h}{\partial r_{a,b}}\Big|  \alpha^ {- \lceil \log_4 (w_a + 1) \rceil} \leq 1\]
for any $d\leq 4$,
$\underline{w} \geq \mathbf{1}$,
and $\mathbf{0} \leq \set{r_{a,b}}_{1\leq a \leq d, 1\leq b \leq w_a} \leq \mathbf{1}$.
There are several difficulties to prove this inequality.
First of all, the variables $\set{r_{a,b}}$ are not totally symmetric: there are $d$ groups with the $a$-th group again has $w_a$ variables. Secondly, the term $\alpha^ {- \lceil \log_4 (w_a + 1) \rceil}$ is a discontinuous function, which behaves quite differently for small or large $w_a$s. To overcome these difficulties and carry out the proof, we introduce some new proof ideas and the following is an outline:
\begin{itemize}
\item In Claim \ref{max_equivalence}, we show that
in the worst case the two-layer decay rate $\hat{\kappa}_d$ is equivalent to a single-layer rate $\kappa_d$,
via Karamata's Inequality~\cite{kadelburg2005inequalities}. Basically, for each group of variables $\set{r_{a,b}}_{b=1}^{w_a}$, the maximum is achieved when all but one variables take the boundary value $1$.
\item We show this artificial single-layer rate is sub-additive in Claim \ref{decomposeKappa}, which enables separating groups with larger $w_a$ from smaller ones, and dealing with them separately.
\item Finally, with careful numerical analytics on small and large groups respectively in Claim \ref{numeric_kappa}, we complete the proof by combining them with Claim \ref{decomposeKappa}.
\end{itemize}

Before detailing the proof, we first do a change of variables to simplify notations. Let $\hat{t}_{j,i} \triangleq \frac{r_{j,i}}{1+r_{j,i}}$, which is just the marginal probability as $r_{j,i}$ is the ratio of  marginal probability.  Then   $0\leq \hat{t}_{j,i} \leq \frac{1}{2}$, function $h$ becomes  $\hat{h}(\mathbf{\hat{t}}) \triangleq \prod_j^d \left( 1 - \prod_i^{w_j} \hat{t}_{j,i} \right)$, and the function $\hat{\kappa}$ becomes
\begin{align*}
\hat{\kappa}_d^{\underline{w}} (\mathbf{\hat{t}})
\triangleq &\sum_{a,b} \frac{\Phi(h(\mathbf{r}))}{\Phi(r_{a,b})} \abs{\frac{\partial h}{\partial r_{a,b}}} \alpha^{- \lceil \log_4 (w_a + 1) \rceil} \\
= &\sqrt{\frac{\hat{h}(\mathbf{\hat{t}})}{1+\hat{h}(\mathbf{\hat{t}})}} \cdot \sum_a^d \frac{\alpha^{- \lceil \log_4 (w_a + 1) \rceil}  \prod_i^{w_a} \hat{t}_{a,i} }{1-\prod_i^{w_a} \hat{t}_{a,i}} \cdot \sum_b^{w_a} \frac{1 - \hat{t}_{a,b}}{\sqrt{\hat{t}_{a,b}}} \end{align*}
where $\underline{w}$ is the ascending ordered sequence $\set{w_a}_{a=1}^d$.
Next we introduce the single-layer rate $\kappa_d$,
\begin{equation}
\kappa_d^{\underline{w}} (\mathbf{t}) \triangleq
\sqrt{\frac{\prod_j^d \left( 1 - \frac{t_j}{2^{w_j - 1}}\right)}{1+\prod_j^d \left( 1 - \frac{t_j}{2^{w_j - 1}}\right)}} \cdot
\sum_a^d \frac{\frac{t_a}{2^{w_a - 1}}}{1-\frac{t_a}{2^{w_a - 1}}} \cdot
\left(\frac{1-t_a}{\sqrt{t_a}} + \frac{w_a-1}{\sqrt{2}}\right) \cdot \alpha^{- \lceil \log_4 (w_a + 1) \rceil}.
\end{equation}

Note that $\kappa_d$ is essentially fixing for each $a$ and each $b>1$, $\hat{t}_{a,b} = \frac{1}{2}$, leaving only $\hat{t}_{a,1}$ free in $\hat{\kappa}_d$ and renamed as $t_a$ .
Clearly
$\max_{\mathbf{t}}\set{\kappa_d^{\underline{w}}(\mathbf{t})} \leq \max_{\hat{\mathbf{t}}}\set{\hat{\kappa}_d^{\underline{w}}(\hat{\mathbf{t}})}$, we prove they are indeed equal.

\begin{claim}
\label{max_equivalence}
$\max_{\mathbf{t}}\set{\kappa_d^{\underline{w}}(\mathbf{t})} = \max_{\hat{\mathbf{t}}}\set{\hat{\kappa}_d^{\underline{w}}(\hat{\mathbf{t}})}$.
\end{claim}
\begin{proof}
We only need to prove that for any $\hat{\mathbf{t}}$, there exist a $\mathbf{t}$ such that $\kappa_d^{\underline{w}}(\mathbf{t}) \geq \hat{\kappa}_d^{\underline{w}}(\hat{\mathbf{t}})$.
For any given $\hat{\mathbf{t}}$, we define $\tilde{\mathbf{t}}$ and $\mathbf{t}$  as follows:  for each $a$, $\tilde{t}_{a,b} = \frac{1}{2}$ for  $b>1$ and
$t_a= \tilde{t}_{a,1} =2^{w_a-1} \prod_{b=1}^{w_a}\hat{t}_{a,b}$.
By definition, we have $\prod_{b=1}^{w_a}\hat{t}_{a,b}= \prod_{b=1}^{w_a}\tilde{t}_{a,b}$ for each $a$ and $\kappa_d^{\underline{w}}(\mathbf{t})= \hat{\kappa}_d^{\underline{w}}(\tilde{\mathbf{t}})$. Thus, it is sufficient to prove that
$\hat{\kappa}_d^{\underline{w}}(\hat{\mathbf{t}}) \leq \hat{\kappa}_d^{\underline{w}}(\tilde{\mathbf{t}})$.
By the expression of $\hat{\kappa}_d^{\underline{w}}$ and the fact that $\prod_{i=1}^{w_a}\hat{t}_{a,i}= \prod_{i=1}^{w_a}\tilde{t}_{a,i}$, we only need to prove that for each $a=1,2, \cdots, d$,
\[\sum_b^{w_a} \frac{1 - \hat{t}_{a,b}}{\sqrt{\hat{t}_{a,b}}} \leq \sum_b^{w_a} \frac{1 - \tilde{t}_{a,b}}{\sqrt{\tilde{t}_{a,b}}}.\]
We shall prove this by Karamata's Inequality~\cite{kadelburg2005inequalities}, which is the opposite direction of Jensen's Inequality. First we do a change of variables. For a fixed $a$, let
$p_{b} \triangleq \ln (\hat{t}_{a,b})$, $q_{b} \triangleq \ln (\tilde{t}_{a,b})$ and $f(x) \triangleq \frac{1-e^x}{\sqrt{e^x}}$,
clearly we have
$\sum_b p_{b} = \ln (\prod_b \hat{t}_{a,b}) = \ln (\prod_b \tilde{t}_{a,b}) = \sum_b q_{b}$, and our goal is
 \[\sum_b^{w_a} \frac{1 - \hat{t}_{a,b}}{\sqrt{\hat{t}_{a,b}}} = \sum_b^{w_a} f(p_{b})\leq  \sum_b^{w_a} f(q_{b}) =\sum_b^{w_a} \frac{1 - \tilde{t}_{a,b}}{\sqrt{\tilde{t}_{a,b}}}.\]
Since the second derivative $f''(x) = \frac{1-e^x}{4 \sqrt{e^x}} > 0$ for $x \leq \ln \frac{1}{2}$,
$f$ is strictly convex. Thus the above inequality
immediately follows from Karamata's Inequality and the fact that the sequence $\set{q_{b}}_{b=1}^{w_a}$ always majorizes $\set{p_{b}}_{b=1}^{w_a}$ after both being reordered in descending order.

%
%
\end{proof}

In light of this, instead of $\hat{\kappa_d}$ we consider the simplified amortized decay rate $\kappa_d$ .
Denote $T(t,k) \triangleq \frac{\frac{t}{2^{k-1}} }{1-\frac{t}{2^{k-1}}}\left(\frac{1-t}{\sqrt{t}} + \frac{k-1}{\sqrt{2}}\right)$,
$\hat{h}(t,k) \triangleq 1-\frac{t}{2^{k-1}}$,
$g(h) \triangleq \sqrt{\frac{h}{1+h}}$,
$\hat{g} \triangleq g(\prod_j^d \hat{h}(t_j,w_j))$.
Clearly $\kappa_d^{\underline{w}}(\mathbf{t})$ can be re-written as
$\kappa_d^{\underline{w}}(\mathbf{t}) = \hat{g}  \cdot \sum_a^d \frac{T(t_a,w_a)}{ \alpha^{ \lceil \log_4 (w_a + 1) \rceil}}$.

\begin{claim}[Sub-additivity]
\label{decomposeKappa}
$\kappa_d^{\underline{w}} (\mathbf{t}) \leq \kappa_{d-1}^{\underline{w}\setminus (w_a)} (\mathbf{t} \setminus \set{t_a}) + \kappa_{1}^{(w_a)}(\set{t_a})$.
\end{claim}

\begin{proof}
$g(h)$ is monotonically increasing, and by $\hat{h}(t_j,w_j)\leq 1$, we have
$\forall a, \prod_j^d \hat{h}(t_j,w_j) \leq \hat{h}(t_a,w_a), \prod_j^d \hat{h}(t_j,w_j) \leq \prod_{j \neq a}^d \hat{h}(t_j,w_j)$,  thus $\forall a$,
$ \hat{g} \leq g(\hat{h}(t_a, w_a)) \textrm{ and }\ \hat{g} \leq g(\prod_{j \neq a}^d \hat{h}(t_j, w_j))$.
Hence
$\hat{g}  \cdot \sum_j^d \frac{T(t_j,w_j)}{\alpha^{ \lceil \log_4 (w_j + 1) \rceil}} \leq g(\hat{h}(t_a,w_a)) \cdot \frac{T(t_a,w_a)}{\alpha^{ \lceil \log_4 (w_a + 1) \rceil}} + g(\prod_{j \neq a}^d \hat{h}(t_j,w_j))\cdot \sum_{j\neq a}^d \frac{T(t_j,w_j)}{\alpha^{ \lceil \log_4 (w_j + 1) \rceil}}$.
\end{proof}

\begin{claim}[Numerical Bounds]
\label{numeric_kappa}
For any $\mathbf{t}$ with $\mathbf{0} \leq \mathbf{t} \leq \mathbf{\frac{1}{2}}$, 
\begin{itemize}
  \item For $\underline{w} < \mathbf{4}$, $\kappa_4^{\underline{w}}(\mathbf{t}) < 1$,
  $\kappa_3^{\underline{w}}(\mathbf{t}) < 0.85$,
  $\kappa_2^{\underline{w}}(\mathbf{t}) < 0.67$,
  $\kappa_1^{\underline{w}}(\mathbf{t}) < 0.42$.
  \item For $w_1 \geq 4$, $\kappa_1^{(w_1)}(\mathbf{t}) < 0.14$.
\end{itemize}
\end{claim}
The proof of this claim is quite complicated and involves many case-by-case analysis, which we defer to the appendix.
%

Now we are ready to complete the rest of the proof for Lemma~\ref{main-lem}.

{\bf\noindent Proof of Lemma~\ref{main-lem}.}
Given the bound (\ref{equ:bound-d4}),
the case where $d_x(C)=5$ immediately follows by Claim \ref{numeric_kappa} that $ \kappa_1^{(w)}(\mathbf{t}) < \frac{1}{2}$, and Claim \ref{decomposeKappa},
\begin{equation}
\label{main-lem-d5}
\abs{ \varphi \circ R(C,x,L) -  \varphi \circ R(C,x)} \leq \kappa_5^{\underline{w}}(\mathbf{t}) \cdot 2 \alpha^L \leq 5 \kappa_1^{(w)}(\mathbf{t}) \cdot 2 \alpha^L \leq 5 \alpha^L.
\end{equation}
It remains to combine the small-degrees and large-degrees of $\underline{w}$ and get $\kappa_d^{\underline{w}} (\mathbf{t}) < 1$ for $d \leq 4$.

We choose $M = 4$, let $\underline{w}^{+}$ be the sub-sequence of those large entries (i.e. any entry greater or equal to $M$) in $\underline{w}$ , and similarly for $\underline{w}^{-}$, or formally,
\[\underline{w}^{+} \triangleq \set{w_j : w_j \geq M} \ \mathrm{  and   }\  \ \underline{w}^{-} \triangleq \set{w_j : w_j < M},\]
\[t^{+} \triangleq \set{t_j : w_j \geq M}\  \mathrm{  and   }\ \ t^{-} \triangleq \set{t_j : w_j < M}.\]

Next we combine them case by case by Claim \ref{numeric_kappa} and \ref{decomposeKappa}.
\begin{itemize}
\item
 No entry of $\underline{w}$ is large, i.e. $\underline{w} < \mathbf{M}$.
We have $\kappa_d^{\underline{w}}(\mathbf{t}) < 1$ for $d=1,2,3,4$;

\item
 All entries of $\underline{w}$ are large, i.e. $\underline{w} \geq \mathbf{M}$.
We have $\kappa_d^{\underline{w}}(\mathbf{t}) \leq  4 \kappa_1^{(w_a^{+})}\left(\set{t_a^{+}}\right) < 4\cdot 0.14 <1$;

\item
 Three entries of $\underline{w}$ are large.
  Since $d=3$ is the same to the above, we assume $d=4$.
$
 \kappa_4^{\underline{w}}(\mathbf{t}) \leq 3 \kappa_1^{(w_a^{+})}\left(\set{t_a^{+}}\right) + \kappa_1^{(w_a^{-})}\left(\set{t_a^{-}}\right) \leq  3\cdot 0.14 + 0.42 < 1
$;

\item
 Two entries of $\underline{w}$ are large.
    $ \kappa_4^{\underline{w}}(\mathbf{t}) \leq
     2\kappa_1^{(w_a^{+})}\left(\set{t_a^{+}}\right) + \kappa_2^{\underline{w}^{-}}(\mathbf{t}^{-}) <
     2\cdot 0.14 + 0.67 < 1$, and
     $ \kappa_3^{\underline{w}}(\mathbf{t}) \leq
      2\kappa_1^{(w_a^{+})}\left(\set{t_a^{+}}\right) + \kappa_1^{\underline{w}^{-}}(\mathbf{t}^{-}) <
     2\cdot 0.14 + 0.42<1$;

\item
 One entry of $\underline{w}$ is large.
     $ \kappa_4^{\underline{w}}(\mathbf{t}) \leq
     \kappa_1^{(w_a^{+})}\left(\set{t_a^{+}}\right) + \kappa_3^{\underline{w}^{-}}(\mathbf{t}^{-}) <
     0.14 + 0.85 < 1$,

     $ \kappa_3^{\underline{w}}(\mathbf{t}) \leq
      \kappa_1^{(w_a^{+})}\left(\set{t_a^{+}}\right) + \kappa_2^{\underline{w}^{-}}(\mathbf{t}^{-}) <
     0.14 + 0.67<1$,

     $ \kappa_2^{\underline{w}}(\mathbf{t}) \leq
      \kappa_1^{(w_a^{+})}\left(\set{t_a^{+}}\right) + \kappa_1^{\underline{w}^{-}}(\mathbf{t}^{-}) <
     0.14 + 0.42<1$.
\end{itemize}
     \qed

\section{Bounded Degree Boolean \#CSP}\label{sec:csp}
In this section, we take a broader view to study the approximability of  $\CSP^c_d(\Gamma)$ in terms of $\Gamma$ and $d$.
For $d\ge 6$, \cite{DyerGJR12} gave
a complete classification in terms of $\Gamma$: the problem is NP-hard under randomized reduction, BIS-hard (as hard as approximately counting independent sets for a bipartite graph), or in FP even for exact counting.
Basically,  there is no interesting approximable cases under these hardness assumptions.
Thus, we focus on these  $d\in \{3,4,5\}$.
In particular, our algorithm for monotone CNF with $d=5$  falls in this range and
gives an interesting approximable family. A partial classification for $d\in \{3,4,5\}$ was also given in~\cite{DyerGJR12}, and the monotone constraints are the only unknown case.

As already discussed in~\cite{DyerGJR12}, any monotone Boolean constraint can be written as a monotone CNF. Specifically, except for a few trivial cases, the expression is unique if only well-formed CNF is used.  They also introduced the notion of variable rank $k$, which is the maximum occurrence in terms of CNF clauses  within a monotone Boolean constraint in $\Gamma$. They then show that
\[\mbox{Read-$d$-Mon-CNF} \leq_{AP} \CSP^c_d(\Gamma) \leq_{AP} \mbox{Read-$k d$-Mon-CNF}.\]
 $A \leq_{AP} B$ is the approximation-preserving reduction, which means that if there is an approximate counting algorithm for problem $B$, then there is also an approximate counting algorithm for problem $A$. Therefore, our above FPTAS for Read-$5$-Mon-CNF made an important step towards a full classification for $\CSP^c_d(\Gamma)$.
 In particular, if $k=1$ we have a complete dichotomy. For $k \geq 2$, there is still gap in between if $d\leq 5$ while $k d \geq 6$,
for which we identify two fundamental families of constraints.
Consider the unique well-formed CNF of such a constraint, there is a variable $x$ that appears in at least two clauses $c_1,c_2$.
Since they are not singleton and no one is a subset of the other, we can find two variables,
$y$ from  $c_1$, $z$ from $c_2$, such that $y$ does not appear in $c_2$ and $z$ does not appear in $c_1$.
For well-formed monotone CNF, there is a suitable pinning for all other variables to isolate these three variables.
After pinning, at least two clauses remain: $x\vee y$ obtained from $c_1$ and $x\vee z$ from $c_2$.
All other possible clauses are $y\vee z$.
If it is not present, denote the constraint by $S_2(x,y,z)= (x\vee y) \wedge (x\vee z)$; if $y\vee z $ is present,
denote it by $K_3(x,y,z)= (x\vee y) \wedge (x\vee z) \wedge (y \vee z)$. These structures can be generalized to
higher arity with larger variable rank as  $S_k(x,y_1,y_2, \cdots, y_k)= \bigwedge_{i=1}^k  (x\vee y_i) $ and $K_s(x_1,x_2, \cdots, x_s)= \bigwedge_{1\leq i< j\leq s} (x_i\vee x_j) $. 
Thus it is crucial to understand the approximability of these two families. We investigate and get the following FPTAS and hardness result.

\begin{theorem}
\label{counting-thm2}
There is an FPTAS for $\CSP_3^c(K_4)$.
\end{theorem}

As a remark, this is a slightly more general problem of 3D matching problem, where non-uniform hyperedges are also allowed.
Again, if we expand out the constraints $K_4$ in a $\CSP_3^c(K_4)$ instance to monotone CNF formulas, we get a Read-$9$-Mon-CNF.
 Although approximately counting Read-$9$-Mon-CNF in general is hard, we get an FPTAS for this sub-family by leveraging the locally clique-like structure of $K_s$.
 Thus, besides the degree $d$, variable rank $k$, one needs to study the inner structure of a monotone constraint $\Gamma$ to determine its approximability.

\begin{theorem}\label{thm-hardness}
There is no FPTAS (or FPRAS) for $\#CSP_5^c(S_2)$ unless $NP=RP$.
\end{theorem}

This hardness result is proved by simulating the reduction for $\#CSP_6(OR_2)$ in~\cite{Sly10}. If the same hardness result is also proved for $\#CSP_5^c(K_3)$, one obtains a complete classification for $\CSP^c_d(\Gamma)$ with $d\geq 5$, based on our above discussion.
However, the same idea to simulate the reduction for $\#CSP_6(OR_2)$ does not work.
Specifically, a bi-partite gadget plays a crucial role for the reduction in~\cite{Sly10} while $K_3$ is inherently not bi-partite.
We leave the full classification for $\CSP^c_d(\Gamma)$ as an interesting and important open question.



\subsection{FPTAS for $\CSP_3^c(K_4)$}
The constraint $K_{d}$ states that at most one variable can be assigned as $0$. Such {\sc At-Most-One} constraint is exactly the matching constraint: every edge $e$ corresponds to a variable $x_e$, and a matching $M$ with $e \in M$ corresponds to an assignment that assigns $x_e=0$, and every vertex with degree $d$ is a constraint $K_{d}$ involving its $d$ neighboring edges.
As a result, $\CSP_2^c(K_{d})$ is essentially counting matchings in graphs of maximum degree $d$.
And $\CSP_3^c(K_4)$ contains $3$ dimensional matching over graphs with maximum degree $4$ as a special case. 

Due to pinning, we have  $\CSP_3^c(K_4) \equiv \CSP_3^c(\{K_4, K_3, K_2, K_1\})$, where $K_1(x)$ is satisfied for both $x=0,1$.
One can always normalize  (in polynomial time) $C$ in $\CSP_3^c(\{K_4, K_3, K_2, K_1\})$ to a new instance without
pinning i.e. $C'$ in $\CSP_3(\{K_4, K_3, K_2, K_1\})$, unless we can easily decide that $C$ is not satisfiable.
First we deal with variables that are pinned to $0$.
If a constraint $K_s$ has at least two variables (including two occurrences of a same variable) pinned to $0$, then it is not satisfiable.
If only $x_i$ in $K_s(x_1,x_2, \cdots, x_s)$ is pinned to $0$, all other variables $\set{x_j}_{j\neq i}^s$ must be $1$.
So we pin them to $1$ and remove this constraint.
If a variable appears more than once in a same constraint, it must be $1$ and we add a pinning to $1$ to it.
Now it only remains to tackle variables that are pinned to $1$. If $x_i$ in $K_s(x_1,x_2, \cdots, x_s)$ is pinned to $1$, we simply replace the constraint with $K_{s-1} (\set{x_j}_{j \neq i}^s)$.
It is obvious that the underlying Boolean function remains unchanged after normalizing, hence so does the number of satisfying assignments.
Since the pinned variables are constants and do not appear in any constraints after normalizing, we simply remove them and no longer see them as variables of the instance.
Also after normalizing there are no duplicate variables in a same constraint.
Now we assume the given instance $C$ is already in $\CSP_3(\{K_4, K_3, K_2, K_1\})$ as we will also perform this process during the algorithm.

%

Overall we carry out a similar scheme as that for monotone CNF: first show the recursion for $\#CSP_d^c(K_s)$ after preprocessing, and prove correlation decay up to $d=3, s=4$.
Here we re-use the same set of notations for the same concept but under the setting of $\#CSP_3^c(K_4)$.
Given a $\#CSP_d^c(K_s)$ instance $C$, we associate the set of satisfying assignments with a uniform distribution, and consider the ratio of marginal probability $R(C,x) \triangleq \frac{\mathbb{P}_C ( x=0)}{\mathbb{P}_C ( x=1)}$.


\begin{lemma}
Let $C$ be a normalized instance of  $\CSP_3(\{K_4, K_3, K_2, K_1\})$ and $d$ constraints containing $x$  be enumerated as $\set{c_j}_{j=1}^d$.
Denote $w_j \triangleq \abs{c_j}-1$, $\set{x_{j,i}}_{i=1}^{w_j}$ as the set of variables in $c_j$ except $x$,
$C_j$ be obtained from $C$ by pinning the occurrences of $x$ in $c_1, c_2, \cdots, c_{j-1}$ to $1$  and the occurrences of $x$ in $c_{j+1}, c_{j+2}, \cdots, c_{d}$ to $0$ ,
$C_{j,i}$ be further obtained from $C_j$ by removing clause $c_j$ and pinning variables  $x_{j,k}$ with $k\neq i$ (all occurrences) to $1$.
%
\[R(C,x) = \prod_{j=1}^{d}  \frac{1}{1+\sum_i^{w_j} R(C_{j,i},x_{j,i})}. \]
\end{lemma}
\begin{proof}
If $d=0$, $x$ is a free variable and $R(C,x) = 1$.
If $d\geq 1$, we substitute the $d$ occurrences of variable $x$ with $d$ independent variables $\set{\tilde{x}_j}_{j=1}^d$, and denote this new $\#CSP_d^c(K_s)$ by $C'$. Then we have
\[
R(C,x)=\frac{\mathbb{P}_C ( x=0)}{\mathbb{P}_C ( x=1)}
= \frac{\mathbb{P}_{C'} \left( \set{\tilde{x}_j} = \mathbf{0} \right)}{\mathbb{P}_{C'} \left( \set{\tilde{x}_j} = \mathbf{1} \right)}
= \prod_{j=1}^d \frac{\mathbb{P}_{C'} \left( \set{\tilde{x}_i}_{i=1}^{j-1} = \mathbf{1},  \set{\tilde{x}_i}_{i=j}^{d} = \mathbf{0} \right)}{\mathbb{P}_{C'} \left(  \set{\tilde{x}_i}_{i=1}^{j} = \mathbf{1},  \set{\tilde{x}_i}_{i=j+1}^{d} = \mathbf{0}  \right)}
= \prod_{j=1}^d R(C_j, \tilde{x}_j). \]
Recall that $c_j$ is the constraint where $\tilde{x}_j$ resides,  $\set{x_{j,i}}_{i=1}^{w_j}$ is the set of variables in $c_j$ except $\tilde{x}_j$,
\begin{align*}
 R(C_j, \tilde{x}_j)
&= \frac{\mathbb{P}_{C_j} \left( \set{x_{j,i}}_{i=1}^{w_j} = \mathbf{1} \right)}{\mathbb{P}_{C_j} \left( \set{x_{j,i}}_{i=1}^{w_j} = \mathbf{1} \right) + \sum_k^{w_j} \mathbb{P}_{C_j} \left( \set{x_{j,i}}_{i : i \neq k}^{w_j} = \mathbf{1}, x_{j,k} = 0 \right)}\\
&= \cfrac{1}{1+ \sum_k^{w_j} \frac{\mathbb{P}_{C_j} \left( \set{x_{j,i}}_{i : i \neq k}^{w_j} = \mathbf{1}, x_{j,k} = 0 \right)}{\mathbb{P}_{C_j} \left( \set{x_{j,i}}_{i=1}^{w_j} = \mathbf{1} \right)}} = \frac{1}{1+\sum_i^{w_j} R(C_{j,i},x_{j,i})}.
\end{align*}
Note that $ R(C_j, \tilde{x}_j)=1$ for $w_j=0$, so the above is still true. Hence by substitution the proof is concluded.
\end{proof}

Similarly, we can conclude that $d_{x_{j,i}} ( C_{j,i} ) \leq 2$ since one occurrence of $x_{j,i}$ has been eliminated. This is still true after normalizing $C_{j,i}$. Since we only pin some occurrences of $x$ to $0$ which is no longer a variable for $C_{j,i}$,
$C_{j,i}$ must be satisfiable. After normalizing, $x_{j,i}$ may be pinned to $1$, in which case we have  $R(C_{j,i},x_{j,i})=0$. Otherwise, $x_{j,i}$ is not pinned in normalized $C_{j,i}$, and we continue to expand these $R(C_{j,i},x_{j,i})$ by the above recursion and get a computation tree for  $R(C,x)$.
Similarly, by truncating the above recursion to depth $L$ we get an algorithm for estimating the ratio of marginal probability with running time $O(6^L)$.
Let $d \triangleq d_x(C)$,
\[R(C,x,L) =
\left\{
  \begin{array}{ll}
    0, & \hbox{\textrm{$x$ is pinned to $1$};} \\
    1, & \hbox{$d=0$ or $L=0$;} \\
    \prod_{j=1}^{d} \frac{1}{1+\sum_i R\left(C_{j,i},\ x_{j,i},\ L - 1\right)}, & \hbox{$1 \leq d \leq 3$.} \\
  \end{array}
\right.
\]
Similar to the CNF problem, we also prove correlation decay and obtain an FPTAS for $\#CSP_3^c(K_4)$.


\begin{lemma}
\label{main-lem2}
Let $\alpha = 0.99$,  $C$ be an instance of $\#CSP_3^c(K_4)$ and $x$ be a variable of $C$.
Then
\begin{equation}
\label{main-lem2-equation}
\abs{ R(C,x,L) -  R(C,x)} \leq 4 \sqrt{6} \alpha^L .
\end{equation}
\end{lemma}

The overall structure to prove this lemma and then Theorem \ref{counting-thm2} is the same as that for Monotone CNF. The key part is the following amortized decay condition and the remaining proof is deferred to appendix.

\begin{claim}[Amortized Decay Condition]
\label{condition-triangle}
Let $d\leq 2$, $ w_j \leq 3$  for $1\leq j \leq d$, $0\leq r_{j,i}\leq 1$ for $1\leq j \leq d$ and $1 \leq i \leq w_j$,
$h(\mathbf{r}) \triangleq \prod_{j=1}^{d} \frac{1}{1+\sum_i^{w_j} r_{j,i}}$, and $\Phi(x)=\frac{1}{\sqrt{x(x+1)}}$. Then
\begin{equation}
\label{condition-triangle-equation}
\hat{\kappa}_d^{\underline{w}} (\mathbf{r}) \triangleq \sum_{j,i} \frac{\Phi(h(\mathbf{r}))}{\Phi(r_{j,i})} \abs{\frac{\partial h}{\partial r_{j,i}}}
= \frac{1}{\sqrt{1+ \prod_k^d(1+\sum_i^{w_k} r_{k,i})}} \sum_j^d \frac{\sum_i^{w_j} \sqrt{r_{j,i}(1+ r_{j,i})}}{(1+ \sum_i^{w_j} r_{j,i}) }
< 0.99.
\end{equation}
\end{claim}

Here we also have a two-layer structure for variables $\set{r_{j,i}}$. Similarly, we first show a reduction to a single-layer decay rate $\kappa_d$ in Claim \ref{single-layer-triangle}. Unlike before, here we use Jensen's Inequality to show that the worse case is achieved when the variables in the same group have the same value. We then also prove sub-additivity in Claim \ref{sub-triangle}.

\begin{claim}
\label{single-layer-triangle}
For every $\mathbf{r'} \triangleq \set{r'_{j,i}}$, there exists $\mathbf{r}  \triangleq \set{r_j}$ such that
\[ \hat{\kappa}_d^{\underline{w}}(\mathbf{r'}) \leq \kappa_d^{\underline{w}} (\mathbf{r})\triangleq \frac{1}{\sqrt{1+ \prod_k^d(1+w_k r_k)}} \sum_j^d \frac{w_j \sqrt{r_j(1+ r_j)}}{(1+ w_j r_j) }.\]
\end{claim}

\begin{proof}
Let $r_j = \frac{1}{w_j} \sum_{k=1}^{w_j} r'_{j,k}$,
by Jensen's Inequality, and $f(x)\triangleq \sqrt{x(1+x)}$ is concave,
\[ \left( \forall j, \sum_{i=1}^{w_j} f(r'_{j,i}) \leq w_j f(r_j) \right) \implies
\hat{\kappa}_d^{\underline{w}}(\mathbf{r'}) \leq
\kappa_d^{\underline{w}}(\mathbf{r}). \]
\end{proof}

\begin{claim}[Sub-additivity]
\label{sub-triangle}
$\kappa_d^{\underline{w}} (\mathbf{r}) \leq \sum_j^d \kappa_1^{(w_j)} (r_j)$.
\end{claim}
\begin{proof}
Since $\forall j, \prod_k^d(1+w_k r_k) \geq 1+ w_j r_j$,
\begin{align*}
 \frac{1}{\sqrt{1+ \prod_k^d(1+w_k r_k)}} \sum_j^d \frac{w_j \sqrt{r_j(1+ r_j)}}{(1+ w_j r_j) } \leq \sum_j^d \frac{1}{\sqrt{1+ (1+w_j r_j)}} \frac{w_j \sqrt{r_j(1+ r_j)}}{(1+ w_j r_j) }.
\end{align*}

\end{proof}

{\bf\noindent Proof of Claim \ref{condition-triangle}.}

By Claim \ref{single-layer-triangle}, it is sufficient to prove $\kappa_d^{\underline{w}} (\mathbf{r})<0.99$ for $d=1, 2$, and $\forall j, w_j \leq 3$.


{\bf\noindent Case $d=1$:}
$\kappa_1$ itself is a function over single variable, by simple calculus we have,
\[ \kappa_1^{(1)}(r) < 0.41,\ \ \kappa_1^{(2)}(r) \leq 0.5,\ \ \kappa_1^{(3)}(r) < 0.58.\]

{\bf\noindent Case $d=2, w_1 = w_2 = c$:}
Let $p_j = \ln(1+ c r_j)$, $p' = \frac{\sum_j^d p_j}{d}$, $f_c(p) \triangleq e^{-p} \sqrt{\left(e^p-1\right) \left(c+e^p-1\right)}$,
$r' = \frac{e^{p'} - 1}{c}$.
By Jensen's Inequality and $f_c(p)$ is concave for $c\geq 1$ and $0 \leq p \leq \ln (1+c)$,
we have

\[ \sum_{j=1}^{d} f(p_j) \leq d f(p') \implies \kappa_d^{(c,c)} (\mathbf{r}) \leq \bar{\kappa}_d^{(c,c)} (r') \triangleq \frac{c d \sqrt{r' (1+r')}}{(1+c r' ) \sqrt{1 + (1+c r')^d}}. \]

Note that $\bar{\kappa}_2^{(c,c)} (r)$ is again a single-variate function, by simple calculus,  $\bar{\kappa}_2^{(3,3)} (r) < 0.98$, $\bar{\kappa}_2^{(2,2)} (r) < 0.83,\ \ \bar{\kappa}_2^{(1,1)} (r) < 0.651$.

{\bf\noindent Case $\underline{w}=(1,2)$:}
   By sub-additivity $ \kappa_2^{(1,2)} (\mathbf{r} ) \leq \kappa_1^{(2)}(r_1) + \kappa_1^{(1)}(r_2)<  0.91$.

{\bf\noindent Case $\underline{w}=(1,3)$:}
   By sub-additivity $ \kappa_2^{(1,3)} (\mathbf{r} ) \leq \kappa_1^{(3)}(r_1) + \kappa_1^{(1)}(r_2)<  0.99$.

{\bf\noindent Case $\underline{w}=(2,3)$:}
  By Karush Kuhn Tucker (KKT) conditions, the maximal value of $\kappa_2^{(2,3)} (\mathbf{r} )$ is attained either on the boundary, or at stationary points.
  Clearly if $r_1=0$ or $r_2=0$ then $\kappa_2^{(2,3)} (\mathbf{r} )$ is reduced to case $d=1$.
  If $r_1=1$, let $f(r) \triangleq \kappa_2^{(2,3)}(\set{1,r})$, we have
  \[f(r) = \frac{3 \sqrt{2}}{4 \sqrt{8 r+5}} + \frac{2 \sqrt{r (r+1)}}{(2 r+1) \sqrt{8 r+5}} \leq \frac{3 \sqrt{2}}{4 \sqrt{5}} + \frac{ \sqrt{8x}}{(2 r+1) \sqrt{8 r+5}}.\]
  Note that $\frac{3 \sqrt{2}}{4 \sqrt{5}} < 0.48$, and $\frac{ \sqrt{8x}}{(2 r+1) \sqrt{8 r+5}}<0.36$ with $r=\frac{1}{32} \left(\sqrt{185}-5\right)$ attaining its maximal.
  In all, on the boundary where $r_1=1$, we have $\kappa_2^{(2,3)}(\set{1,r}) < 0.84$.

  Now consider stationary points, the partial derivative with respect to $r_1$ is:
\begin{align*}
  \frac{\partial \kappa_2^{(2,3)}(\set{r_1,r_2})}{\partial r_1} =
  \frac{1}{2 \sqrt{r_1 (r_1+1)} (3 r_1+1)^2 (6 r_1 r_2+3 r_1+2 r_2+2)^{3/2}}\left(6+ 6 r_2 \right. \\
 \left. -6 \sqrt{r_1 (r_1+1)} \sqrt{r_2 (r_2+1)} -3 r_1 \left(12 \sqrt{r_1 (r_1+1)} \sqrt{r_2 (r_2+1)}+2 r_2+2 \right. \right. \\
 \left. \left. +3 r_1 \left(r_1 (6 r_2+3)+6 \sqrt{r_1 (r_1+1)} \sqrt{r_2 (r_2+1)}+10 r_2+5\right)\right)\right).
\end{align*}
  Let the partial derivative be zero and solve the quadratic equation on $r_2$,
\begin{align*}
r_2=\frac{1}{-2 (3 r_1+1)^2 \left(9 r_1^3+3 r_1^2-9 r_1+1\right)} \left( 81 r_1^5+90 r_1^4-48 r_1^3-56 r_1^2-5 r_1+2 \right. \\
\left. \pm \sqrt{-r_1^2 (3 r_1+1)^4 \left(81 r_1^5+162 r_1^4+27 r_1^3-97 r_1^2-44 r_1-1\right)} \right).
\end{align*}
Note that only $0 \leq r_2 \leq 1$ is meaningful, we get $r_2'$:
\begin{align*}
r_2'=\frac{1}{-2 (3 r_1+1)^2 \left(9 r_1^3+3 r_1^2-9 r_1+1\right)} \left( 81 r_1^5+90 r_1^4-48 r_1^3-56 r_1^2-5 r_1+2 \right. \\
\left. + \sqrt{-r_1^2 (3 r_1+1)^4 \left(81 r_1^5+162 r_1^4+27 r_1^3-97 r_1^2-44 r_1-1\right)} \right).
\end{align*}

Hence again we have a single-variable function on $r_1$,
or formally $f(r_1)\triangleq \kappa_2^{(2,3)}(\set{r_1,r_2'})$.

Then it is routine to verify that $\D{f(r)}{r}=0$ has exactly two roots at $(0.18,0.19)$ and $(0.45,0.46)$ respectively,
and the maximal is attained at the former one and $\kappa_2^{(2,3)} (\mathbf{r} ) \leq f(r) <  0.91$.
\qed


\subsection{Hardness of Approximating $\#CSP_5^c(S_2)$}
\tikzset{
    photon/.style={decorate, decoration={snake}, draw=Orchid},
    electron/.style={draw=BurntOrange, postaction={decorate},
        decoration={markings,mark=at position .55 with {\arrow[draw=BurntOrange]{>}}}},
    gluon/.style={decorate, draw=magenta,
        decoration={coil,amplitude=2pt, segment length=5pt}},
    photonprime/.style={decorate, decoration={snake}, draw=black},
    electronprime/.style={draw=LimeGreen, postaction={decorate},
        decoration={markings,mark=at position .55 with {\arrow[draw=LimeGreen]{>}}}},
    gluonprime/.style={decorate, draw=cyan,
        decoration={coil,amplitude=2pt, segment length=5pt}},
    myzzz/.style={draw=black, postaction={decorate},
        decoration={markings,mark=at position .55 with {\arrow[draw=black]{>}}}},
}
\definecolor{myblue}{RGB}{80,80,160}
\definecolor{mygreen}{RGB}{80,160,80}

It was shown in~\cite[Theorem 2]{Sly10} that counting independent sets for graphs with maximum degree $6$ (equivalently $\CSP_6(OR_2)$ in our notation) does not admit FPTAS unless $NP=RP$, with a reduction from MAX-CUT.
To prove Theorem \ref{thm-hardness}, we simulate the hard instances of $\CSP_6^c(OR_2)$ used in the reduction in~\cite{Sly10} by instances of $\#CSP_5^c(S_2)$. By pinning, we can realize the constraint $OR_2$ by $S_2$. So it suffices to reduce the variable occurrence from $6$ to $5$. The idea is simple and straightforward.
Intuitively, if variable $x$ appears at two $OR_2$ constraints $OR(x,y)$ and $OR(x,z)$, we replace them by $S_2(x,y,z)=(x\vee y) \wedge (x\vee z)$.
This does not change the underlying constraint nor the number of satisfying assignments but reduce the occurrence of variable $x$ by $1$.
Thus, we only need to show that we can group some pairs of adjacent $OR_2$ constraints into one $S_2$ so that each variable appears in at most five constraints after merging.

To this end, we first review the reduction used in~\cite{Sly10}.
They first construct a bi-partite gadget $G$.
It begins with $5$ random perfect matchings between $W^{+} \cup U^{+}$ and $W^{-} \cup U^{-}$,
and $1$ random perfect matching between $W^{+}$ and $W^{-}$.
Then, they adjoin a collection of $5$-ary trees to $U^{+}$ by identifying the leaves to vertices in $U^{+}$,
and the set of roots becomes $V^{+}$ with degree $5$, do similar operation for $U^{-}$ to get $V^{-}$.
%

Let $H$ be the input graph in the MAX-CUT problem and $G$ be a random bi-partite graph constructed as above. They define $H^{G}$ by first taking $\abs{H}$ copies of $G$ as $\set{G_x}_{x\in H}$, then connect edges between $V_x^{+}$ and $V_y^{+}$ for every edge $(x,y)$ in $H$, and similarly for $V_x^{-}$ and $V_y^{-}$.
This $H^{G}$ is the hardness instance they used in the reduction, and in the language of $\CSP_6(OR_2)$, vertices are variables and edges are $OR_2$ constraints.
%

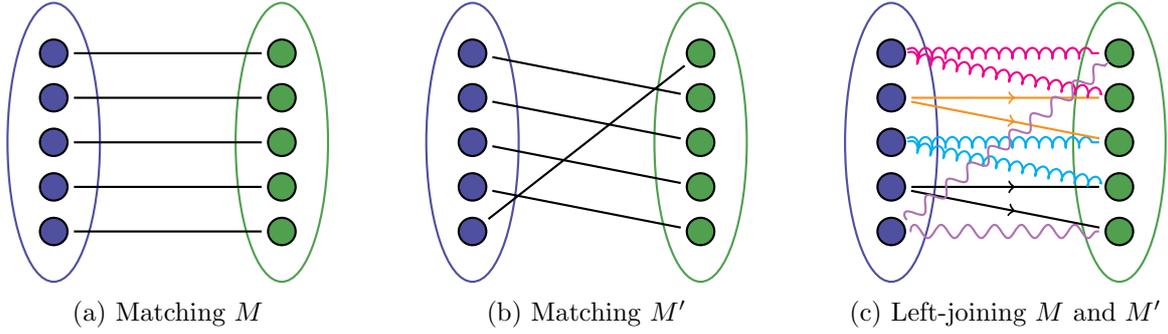
\begin{figure}[htp]

    \begin{subfigure}[b]{0.3\textwidth}
        \centering
        \begin{tikzpicture}[thick,
          every node/.style={draw,circle},
          fsnode/.style={fill=myblue},
          ssnode/.style={fill=mygreen},
          every fit/.style={ellipse,draw,inner sep=-2pt,text width=1cm},
          -,shorten >= 2pt,shorten <= 2pt
        ]
            \begin{scope}[start chain=going below,node distance=2mm]
            \foreach \i in {1,2,...,5}
              \node[fsnode,on chain] (f\i)  {};
            \end{scope}

            \begin{scope}[xshift=3cm,yshift=0cm,start chain=going below,node distance=2mm]
            \foreach \i in {6,7,...,10}
              \node[ssnode,on chain] (s\i) {};
            \end{scope}

            \node [myblue,fit=(f1) (f5)] {};
            \node [mygreen,fit=(s6) (s10)] {};

            \draw (f1) -- (s6);
            \draw (f2) -- (s7);
            \draw (f3) -- (s8);
            \draw (f4) -- (s9);
            \draw (f5) -- (s10);
        \end{tikzpicture}
        \caption{Matching $M$}
    \end{subfigure}
    \hfill
    \begin{subfigure}[b]{0.3\textwidth}
        \centering
        \begin{tikzpicture}[thick,
          every node/.style={draw,circle},
          fsnode/.style={fill=myblue},
          ssnode/.style={fill=mygreen},
          every fit/.style={ellipse,draw,inner sep=-2pt,text width=1cm},
          -,shorten >= 2pt,shorten <= 2pt
        ]
            \begin{scope}[start chain=going below,node distance=2mm]
            \foreach \i in {1,2,...,5}
              \node[fsnode,on chain] (f\i)  {};
            \end{scope}

            \begin{scope}[xshift=3cm,yshift=0cm,start chain=going below,node distance=2mm]
            \foreach \i in {6,7,...,10}
              \node[ssnode,on chain] (s\i) {};
            \end{scope}

            \node [myblue,fit=(f1) (f5)] {};
            \node [mygreen,fit=(s6) (s10)] {};

            \draw (f1) -- (s7);
            \draw (f2) -- (s8);
            \draw (f3) -- (s9);
            \draw (f4) -- (s10);
            \draw (f5) -- (s6);
        \end{tikzpicture}
        \caption{Matching $M'$}
    \end{subfigure}
    \hfill
    \begin{subfigure}[b]{0.3\textwidth}
        \centering
        \begin{tikzpicture}[thick,
          every node/.style={draw,circle},
          fsnode/.style={fill=myblue},
          ssnode/.style={fill=mygreen},
          every fit/.style={ellipse,draw,inner sep=-2pt,text width=1cm},
          -,shorten >= 2pt,shorten <= 2pt
        ]
            \begin{scope}[start chain=going below,node distance=2mm]
            \foreach \i in {1,2,...,5}
              \node[fsnode,on chain] (f\i)  {};
            \end{scope}

            \begin{scope}[xshift=3cm,yshift=0cm,start chain=going below,node distance=2mm]
            \foreach \i in {6,7,...,10}
              \node[ssnode,on chain] (s\i) {};
            \end{scope}

            \node [myblue,fit=(f1) (f5)] {};
            \node [mygreen,fit=(s6) (s10)] {};

            \draw [gluon] (f1) -- (s6);
            \draw [electron](f2) -- (s7);
            \draw [gluonprime] (f3) -- (s8);
            \draw [myzzz] (f4) -- (s9);
            \draw [photon] (f5) -- (s10);

            \draw [gluon] (f1) -- (s7);
            \draw [electron] (f2) -- (s8);
            \draw [gluonprime] (f3) -- (s9);
            \draw [myzzz] (f4) -- (s10);
            \draw [photon] (f5) -- (s6);
        \end{tikzpicture}
        \caption{Left-joining $M$ and $M'$}
    \end{subfigure}
    \hfill
    \caption{Joining matchings on bipartite graph, a color(shape) for a constraint.}
    \label{fig:matching}
\end{figure}

Now we are ready for the simulation. We begin with the construction for bi-partite graph $G$.
Given $5$ random perfect matchings, instead of associating $5$ edges to each vertex, we associate each variable with no more than $4$ constraints as follows:
\begin{itemize}
\item The first $2$ random perfect matchings $M,M'$ is left joined on $W^{+} \cup U^{+}$. Formally for $v \in W^{+} \cup U^{+}$, $(v,u_1) \in M$, $(v,u_2) \in M'$, we replace the two $OR$ constraints with a single $S_2(v,u_1,u_2)$.
    As illustrated in Figure \ref{fig:matching}, this saves one occurrence for variables in $W^{+} \cup U^{+}$ .
\item We take the next $2$ random perfect matchings and group them on $W^{-} \cup U^{-}$. By that, we also save one occurrence for variables in  $W^{-} \cup U^{-}$.
\item For the last random perfect matching $M$, we add $OR(u,v)$ (or $S_2(u, v, 1)$, where the last variable is pinned to $1$) for every $(u,v) \in M$.
\end{itemize}
Next take one more random matching $M$ for $W^{+}$ and $W^{-}$, add $OR(u,v)$ for every $(u,v) \in M$. By these simulation, we save one occurrence for every variable in  $W^{+} \cup U^{+}\cup W^{-} \cup U^{-}$.

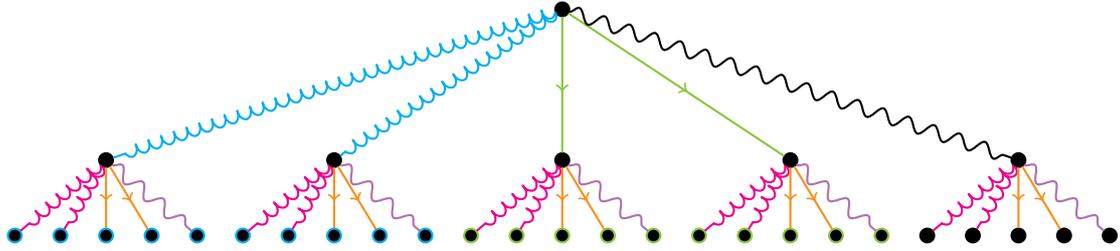
\begin{figure}[htp]

\centering
\begin{tikzpicture}[
        thick,
        level/.style={level distance=20mm/#1, sibling distance=3cm},
        level 2/.style={sibling distance=0.6cm}
    ]
    \node [circle,draw,scale=0.5,fill]{}
        child {
         node[circle,draw,scale=0.5,fill]{}
            child {
                node[circle,draw,scale=0.5,fill] {}
                edge from parent [gluon]
            }
            child {
                node[circle,draw,scale=0.5,fill]{}
                edge from parent [gluon]
            }
             child {
                node[circle,draw,scale=0.5,fill]{}
                edge from parent [electron]
            }
             child {
                node[circle,draw,scale=0.5,fill]{}
                edge from parent [electron]
            }
             child {
                node[circle,draw,scale=0.5,fill]{}
                edge from parent [photon]
            }
            edge from parent [gluonprime]
        }
        child {
         node[circle,draw,scale=0.5,fill]{}
            child {
                node[circle,draw,scale=0.5,fill] {}
                edge from parent [gluon]
            }
            child {
                node[circle,draw,scale=0.5,fill]{}
                edge from parent [gluon]
            }
             child {
                node[circle,draw,scale=0.5,fill]{}
                edge from parent [electron]
            }
             child {
                node[circle,draw,scale=0.5,fill]{}
                edge from parent [electron]
            }
             child {
                node[circle,draw,scale=0.5,fill]{}
                edge from parent [photon]
            }
            edge from parent [gluonprime]
        }
        child {
         node[circle,draw,scale=0.5,fill]{}
            child {
                node[circle,draw,scale=0.5,fill] {}
                edge from parent [gluon]
            }
            child {
                node[circle,draw,scale=0.5,fill]{}
                edge from parent [gluon]
            }
             child {
                node[circle,draw,scale=0.5,fill]{}
                edge from parent [electron]
            }
             child {
                node[circle,draw,scale=0.5,fill]{}
                edge from parent [electron]
            }
             child {
                node[circle,draw,scale=0.5,fill]{}
                edge from parent [photon]
            }
            edge from parent [electronprime]
        }
        child {
         node[circle,draw,scale=0.5,fill]{}
            child {
                node[circle,draw,scale=0.5,fill] {}
                edge from parent [gluon]
            }
            child {
                node[circle,draw,scale=0.5,fill]{}
                edge from parent [gluon]
            }
             child {
                node[circle,draw,scale=0.5,fill]{}
                edge from parent [electron]
            }
             child {
                node[circle,draw,scale=0.5,fill]{}
                edge from parent [electron]
            }
             child {
                node[circle,draw,scale=0.5,fill]{}
                edge from parent [photon]
            }
            edge from parent [electronprime]
        }
        child {
         node[circle,draw,scale=0.5,fill]{}
            child {
                node[circle,draw,scale=0.5,fill] {}
                edge from parent [gluon]
            }
            child {
                node[circle,draw,scale=0.5,fill]{}
                edge from parent [gluon]
            }
             child {
                node[circle,draw,scale=0.5,fill]{}
                edge from parent [electron]
            }
             child {
                node[circle,draw,scale=0.5,fill]{}
                edge from parent [electron]
            }
             child {
                node[circle,draw,scale=0.5,fill]{}
                edge from parent [photon]
            }
            edge from parent [photonprime]
        };

\end{tikzpicture}
\caption{Expressing a $5$-ary tree in $\CSP_4^c(S_2)$, a color(shape) for a constraint.}
\label{fig:tree}
\end{figure}

Finally we adjoin trees onto $U^{+}$ and $U^{-}$ to finish the construction of $G$.
We can group pairs of two $OR_2$ into one $S_2$ constraint in every $5$-ary trees as depicted in Figure \ref{fig:tree}.
By this, we save two occurrences for every variable in those $5$-ary trees except the leaves.
The leaves will be identified with nodes in $U^{+} \cup U^{-}$, for which we have already saved one occurrence.

For the last step to construct $H^G$, we simply use the constraint $OR_2$ and do not do any merging. As the last step does not introduce new nodes and we have saved at least one occurrence for all the nodes for the instance. We conclude that it is an instance for  $\CSP_5^c(S_2)$. This completes the proof.



\bibliographystyle{plain}
\bibliography{refs}

\appendix

\section{Counting from Ratio of Marginal Probability}

First we will be more specific for Theorem \ref{counting-thm2}. For $0< \eps < 1$ and $\#CSP_3^c(K_4)$ problem $C$ with $n$ variables enumerated as $\set{x_i}$, let  $\alpha = 0.99, L = \log_{\alpha} \left(\frac{\eps}{8 \sqrt{6} n} \right)$, $C_1 = C$ and $C_{i+1}$ be from $C_i$ by pinning $x_i$ to $1$. Then
$Z(C, L) \triangleq \prod_i \left( 1 + R(C_i, x_i, L) \right) $ is the desired FPTAS with running time $O\left(n \left(\frac{n}{\eps}\right)^{\log_{1/ \alpha} (6)} \right)$.

%

{\bf\noindent Proof of Theorem \ref{counting-thm1} and \ref{counting-thm2}.}
First note that,
	\[\frac{Z(C,L)}{Z(C)} = \prod_i^n \frac{ 1 + R(C_i, x_i, L) }{ 1 + R(C_i, x_i) }. \]
	By Lemma \ref{main-lem} (resp. Lemma \ref{main-lem2}),
	\[\frac{ \abs{R(C_i,x_i, L) - R(C_i, x_i)} }{ 1 + R(C_i, x_i) } \leq\abs{R(C_i,x_i, L) - R(C_i, x_i)}\leq \frac{\eps}{2n}.\]
	Namely for every $i$,
	\[\left( 1 - \frac{\eps}{2n} \right) \leq \frac{ 1 + R(C_i, x_i, L) }{ 1 + R(C_i, x_i) } \leq \left( 1 + \frac{\eps}{2n} \right), \]
	Hence,
	\[\left( 1 - \frac{\eps}{2n} \right)^n \leq \prod_i^n \frac{ 1 + R(C_i, x_i, L) }{ 1 + R(C_i, x_i) } \leq \left( 1 + \frac{\eps}{2n} \right)^n. \]
	\[1 - \eps \leq \frac{Z(C,L)}{Z(C)} \leq 1+ \eps.\]
	The running time follows from that of $R(C,x,L)$ and there are $O(n)$ calls to it.
\qed

\section{Proof of Claim \ref{numeric_kappa}}

We first introduce a few useful propositions.

\begin{proposition}
\label{equalKappa}
If $w_i = w_j$ and $w_i \leq 2, w_j \leq 2$, we have $\kappa_d^{\underline{w}}(\mathbf{t}) \leq \kappa_d^{\underline{w}}(\mathbf{t'})$ where $t'_i = t'_j = 2^{w_i - 1} \left(1 - \sqrt{\hat{h} (t_i, w_i) \hat{h} (t_j, w_j)} \right)$.
\end{proposition}

\begin{proof}
Let $p_i = \ln(\hat{h} (t_i, w_i))$,
then $t_i = (1-e^{p_i})2^{w_i -1}, \prod_i \hat{h} (t_i, w_i) = e^{\sum_i p_i}, 0 < e^{p_i} \leq 1$, also denote $f(p_i) \triangleq T\left((1-e^{p_i})2^{w_i -1}, w_i\right)$.

If $w_i = w_j = 1$, $f(x) = \sqrt{1-e^x}$, and its second derivative $f''(x) = \frac{e^x \left(e^x-2\right)}{4 \left(1-e^x\right)^{3/2}} <0$, by Jensen's Inequality $f(p_i)+f(p_j) \leq 2 f(\frac{p_i + p_j}{2})$.

If $w_i = w_j = 2$, $f(x) = e^{-x} \left(1-e^x\right) \left(\frac{1-2 \left(1-e^x\right)}{\sqrt{2} \sqrt{1-e^x}}+\frac{1}{\sqrt{2}}\right)$, and
\[f''(x) = - \frac{\left(1 - \sqrt{1-e^x}\right)^2 \left(4 \left(1-e^x\right)^{3/2}+2 \left(1-e^x\right)^2+5 \left(1-e^x\right)+2 \sqrt{1-e^x}+1\right)}{4 e^x \sqrt{2-2 e^x} \left(1-e^x\right)} < 0.\]
also by Jensen's Inequality $f(p_i)+f(p_j) \leq 2 f(\frac{p_i + p_j}{2})$, hence concludes the proof.
\end{proof}

\begin{proposition}[Monotonicity of Partial Derivative]
\label{monotonicity_deriv}
For $\mathbf{t}: \mathbf{0} \leq \mathbf{t} \leq \mathbf{\frac{1}{2}}$, $\underline{w} < \mathbf{4}$, $\frac{\partial \kappa_d^{\underline{w}}(\mathbf{t})}{\partial t_j} $ is monotonically increasing in $w_i$ for $\forall i\leq d, i \neq j$;  while monotonically decreasing in $t_i, \forall i\leq d$.

Furthermore, if $P$ variables $\set{t_p}$ with equal $w_p$ amongst the $d$ components of $\mathbf{t}$ are set to be equal, say $\bar{t}_p$, the above monotonicity preserves for $\frac{\partial \kappa_d^{\underline{w}}(\mathbf{t})}{\partial \bar{t}_p}$.
\end{proposition}
\begin{proof}
For $\underline{w} < \mathbf{4}$, $\alpha^{- \lceil \log_4 (w_a + 1) \rceil} = \alpha^{-1}$ for every $a$, which is just a constant multiplier.
	\[\frac{\partial \kappa_d^{\underline{w}}(\mathbf{t})}{\partial t_j} = \alpha^{-1}\hat{g} \cdot \left(\frac{\partial T(t_j,w_j)}{\partial t_j} - \frac{1}{2^{w_j}} \cdot \frac{1}{\hat{h}(t_j,w_j)(1+\prod_a \hat{h}(t_a,w_a))}  \cdot \sum_a^d T(t_a,w_a) \right).\]
Note that $\hat{h}(t, d)$ is monotonically increasing in $d$ while monotonically decreasing in $t$, hence $\hat{g}$ is monotonically increasing in $w_i, \forall i\leq d$ while monotonically decreasing in $t_i, \forall i\leq d$;
$T(t,d)$ is monotonically increasing in $t$, next we show that $T(t,d)$ is also monotonically decreasing in $d$.

Firstly $T(t,2) - T(t,1) = \frac{\left(\sqrt{2t}-2\right) \sqrt{t}}{2 (2-t)} < 0$ for $t\leq \frac{1}{2}$ ;
Next note that $T(t,w) = \frac{\frac{t}{2^{w-1}} }{1-\frac{t}{2^{w-1}}} \frac{1-t}{\sqrt{t}} + \frac{\frac{t}{\sqrt{2}} }{1-\frac{t}{2^{w-1}}} \frac{w-1}{2^{w-1}}$
and for $w\geq 2$ both parts are monotonically decreasing in $w$.

Now it only remains to show that $\frac{\partial T(t_j,w_j)}{\partial t_j}$ is monotonically decreasing in $t$, since it does not involve $w_i$ for $i \neq j$.

\[\frac{\partial}{\partial t} \left(\frac{\partial T(t,w)}{\partial t} \right)= \frac{-2^{w+\frac{7}{2}} (w-1) t^{3/2}+12 t^2 \left(2^w-1\right)+3 t\cdot 2^w \left(2^w-4\right)+4^w-4 t^3}{2 t^{3/2} \left(2 t-2^w\right)^3}.\]

Let $f(t,w) \triangleq -2^{w+\frac{7}{2}} (w-1) t^{3/2}+12 t^2 \left(2^w-1\right)+3 t\cdot 2^w \left(2^w-4\right)+4^w-4 t^3$, since for $0\leq t \leq \frac{1}{2}$, $2 t-2^w < 0 $ so it suffices to show $f(t,w)\geq 0$. Note that $f(t,w) \geq 2^w\left(2^w-4 (w-1)\right)+\left(3\cdot 2^w \left(2^w-4\right)+12 \left(2^w-1\right) t-  4 t^2\right)t$, and it is easy to check that $f(t,1) = 4(1-t)^3 > 0, f(t,2) = -32 \sqrt{2} t^{3/2}-4 t^3+36 t^2+16 \geq t^2(36 - 4t)>0$, $f(t,w)\geq 0$ for $w\geq 3$ immediately follows from the monotonicity of $2^w-4 (w-1)$ and of the parabola $3\cdot 2^w \left(2^w-4\right)+12 \left(2^w-1\right) t-  4 t^2$. In all we have $\frac{\partial}{\partial t} \left(\frac{\partial T(t,w)}{\partial t} \right) \leq 0 $.

To conclude, by combining the monotonicity of $\hat{g}$, $\frac{\partial T(t,d)}{\partial t}$ and $- \frac{1}{2^{w_j-1}} \cdot \frac{1}{\hat{h}(t_j,w_j)(1+\prod_a \hat{h}(t_a,w_a))}  \cdot \sum_a^d T(t_a,w_a)$, yields the desired monotonicity of $\frac{\partial \kappa_d^{\underline{w}}(\mathbf{t})}{\partial t_j}$.

Furthermore, if we set $P$ variables $\set{t_p}$ with equal $w_p$ to be equal, denote this new set of variables $\mathbf{\bar{t}}$, and denote the variable being set equal as $\bar{t_p}$,

\[\frac{\partial \kappa_d^{\underline{w}}(\mathbf{\bar{t}})}{\partial \bar{t_p}} = P \cdot \alpha^{-1}\hat{g} \cdot \left(\frac{\partial T(\bar{t_p},w_p)}{\partial \bar{t_p}} - \frac{1}{2^{w_p}} \cdot \frac{1}{\hat{h}(\bar{t_p},w_p)(1+\prod_a \hat{h}(t_a,w_a))}  \cdot \sum_a^d T(t_a,w_a) \right).\]

Hence the monotonicity is preserved.
\end{proof}

Now we are ready for the numerical bounds.

{\bf\noindent Proof of Claim \ref{numeric_kappa}.}

Denote $\Lambda_1(w_1) = \max_{\mathbf{t}}\set{\kappa_1^{\underline{w}}(\mathbf{t})},
\Lambda_2(w_1, w_2) = \max_{\mathbf{t}}\set{\kappa_2^{\underline{w}}(\mathbf{t})},
 \Lambda_3(w_1, w_2,w_3) = \max_{\mathbf{t}}\set{\kappa_3^{\underline{w}}(\mathbf{t})},$
 $\Lambda_4\left(w_1, w_2,w_3,w_4\right) = \max_{\mathbf{t}}\set{\kappa_4^{\underline{w}}(\mathbf{t})}$.

\begin{description}
\item[{\bf Case $d=1$.}]

First consider $w_1<4$, by Proposition \ref{monotonicity_deriv} and $\frac{\partial \kappa_1^{(1)}(\set{t})}{\partial t} \big|_{t= \frac{1}{2}} = \frac{1}{3 \sqrt{6} \alpha} > 0$,
to maximize $\kappa_1^{(1)}(\mathbf{t})$ ,$ \mathbf{t}=\mathbf{\frac{1}{2}}$,
hence $\Lambda_1(1) = \frac{1}{\sqrt{6}\alpha}$, similarly to maximize $\kappa_1^{(2)}(\mathbf{t})$ or $\kappa_1^{(3)}(\mathbf{t})$, $ \mathbf{t}=\mathbf{\frac{1}{2}}$, we have $\Lambda_1(2) =\alpha^{-1}\sqrt{\frac{2}{21}} , \Lambda_1(3)=\alpha^{-1}\sqrt{\frac{3}{70}}$.
 In all $\Lambda_1(w_1) \leq \frac{1}{\sqrt{6}\alpha} \approx 0.4162< 0.42$.

 Next for $w_1\geq 4$ note that $\Lambda_1(w_1) \leq \frac{w_1}{2(2^{w_1}-1)\alpha^{\log_4(w_1 + 1)+1}}$, let $f(w_1) \triangleq \frac{w_1}{2(2^{w_1}-1)\alpha^{\log_4(w_1 + 1)+1}}$, since $f(4)< 0.14$,
t is sufficient to check that $\D{f(w_1)}{w_1} < 0$ for $w_1 \geq 4$.
\begin{align*}
\D{f(w_1)}{w_1} =
-&\frac{\left(\frac{125}{109}\right)^{\log_4 (4 w_1+4)} 2^{3 \log_4 (w_1+1)+2} 3^{-2 \log_4 (4k+4))}}{\left(2^{w_1}-1\right)^2 (w_1+1) \log (4)} \\
&\left( 2^{w_1} (w_1-2) \log (2)+2^{w_1} w_1 \log (2) (w_1 \log (4)-4)+ \right. \\
& \left. w_1 \left(2^{w_1} \log (2) \log (4)+2^{w_1} \log \left(\frac{981}{500}\right)+  \log \left(\frac{4000}{981}\right)\right) \right).
\end{align*}
Since $w_1\geq 4$, $w_1 \log (4)-4 > 0$, so $\D{f(w_1)}{w_1} < 0$, hence $\Lambda_1(w_1) < 0.14$ for $w_1\geq 4$.

\item[{\bf Case $d=2$.}]
Since $\Lambda_1(2) =\sqrt{\frac{2}{21}}\alpha^{-1} < 0.31\alpha^{-1}, \Lambda_1(3) = \sqrt{\frac{3}{70}}\alpha^{-1} <0.21\alpha^{-1}$, by Claim \ref{decomposeKappa}, if $w_1 \geq 2$ ,  $\Lambda_2(w_1,w_2) \leq \Lambda_1(w_1) + \Lambda_1(w_2) \leq 2 \Lambda_1(2) < 0.67$;
If $w_2 \geq 3$, $\Lambda_2(w_1,w_2) \leq \Lambda_1(1)+\Lambda_1(3) < 0.651$. So only $(w_1, w_2) = (1,1)$ and $(1,2)$ remains.
We let $\kappa_{2,(1,1)}'(t_1) \triangleq \frac{\partial \kappa_2^{(1,1)}(\set{t_1,t_1})}{\partial t_1}$,
$\kappa_{2,(1,2)}'(t_1) \triangleq \frac{\partial \kappa_2^{(1,2)}(\set{t_1,\frac{1}{2}})}{\partial t_1}$, for the derivative of a single variable function.

\begin{description}
\item[{\bf Case $(1,1)$.}]
 By Proposition \ref{equalKappa}, $\Lambda_2(1,1) = \max_{t_1}\set{\kappa_2^{(1,1)}(\set{t_1,t_1})}$ , so essentially there is only one variable left, and by Proposition \ref{monotonicity_deriv}, fixing $2$ variables into a single variable, say $t_1$, the monotonicity of the derivative of $\kappa$ with respect to $t_1$ is preserved, which is essentially $\kappa_{2,(1,1)}'(t_1)$.

Since $\kappa_{2,(1,1)}'(0.4039) >0, \kappa_{2,(1,1)}'(0.404) <0$, by monotonicity of the partial derivative, $\Lambda_2(1,1) \leq 2 g\left( ( 1- 0.4039)^2 \right) \cdot T(0.404,1) \alpha^{-1} < 0.67$.

\item[{\bf Case $(1,2)$.}]
Since $\frac{\partial \kappa_2^{(1,2)}(\mathbf{t})}{\partial t_2} \big|_{\mathbf{t} =\mathbf{\frac{1}{2}}} = \frac{37}{33 \sqrt{66}\alpha} > 0$, by Proposition \ref{monotonicity_deriv},
to maximize $\kappa_2^{(1,2)}(\mathbf{t})$, $t_2=\frac{1}{2}$. Clearly the partial derivative w.r.t $t_1$ after fixing $t_2 = \frac{1}{2}$ is exactly $\kappa_{2,(1,2)}'(t_1)$, so
by monotonicity and $\kappa_{2,(1,2)}'(0.4533) >0, \kappa_{2,(1,1)}'(0.4534) <0$, we have $\Lambda_2(1,1) \leq g\left( (1 - 0.4533)(1 - \frac{1}{2^2}) \right) \cdot \left( T(0.4534,1) + T(\frac{1}{2},2) \right) \alpha^{-1} < 0.67$.
\end{description}

\item[{\bf Case $d=3$.}]
Note that if $w_j > 1$, in order to maximize $\kappa_3^{\underline{w}}(\mathbf{t})$, $t_j = \frac{1}{2}$. This follows directly from
$\frac{\partial \kappa_3^{(1,1,2)}(\mathbf{t})}{\partial t_3} \big|_{\mathbf{t}=\mathbf{\frac{1}{2}}} = \frac{5}{57 \sqrt{114}\alpha} > 0, \frac{\partial \kappa_3^{(1,1,3)}(\mathbf{t})}{\partial t_3} \big|_{\mathbf{t}=\mathbf{\frac{1}{2}}} = \frac{509}{273 \sqrt{546}\alpha} > 0$ and Proposition \ref{monotonicity_deriv}.

\begin{description}

\item[{\bf Case $(1,1,1)$.}]
 Similar to \emph{Case $d=2,(1,1)$}.Let $\kappa_{3,(1,1,1)}'(t_1) \triangleq \frac{\partial \kappa_3^{(1,1,1)}(\set{t_1,t_1,t_1})}{\partial t_1}$,

  By Proposition \ref{equalKappa}, since $\kappa_{3,(1,1,1)}'(0.3074) >0, \kappa_{3,(1,1,1)}'(0.3075) <0$,
  next by Proposition \ref{monotonicity_deriv}, $\Lambda_3(1,1,1) \leq 3 g\left( ( 1- 0.3074)^3 \right) \cdot T(0.3075,1) \alpha^{-1} < 0.8471$.

\item[{\bf Case $w_1 \geq 2$.}]
In these cases we have $\mathbf{t} = \mathbf{\frac{1}{2}}$, so by direct evaluation we have

$\Lambda_3(2,2,2) = 3 \sqrt{\frac{6}{91}}\alpha^{-1} < 0.786$,

$\Lambda_3(2,2,3) = \frac{37}{\sqrt{2674}}\alpha^{-1} < 0.73$,

$\Lambda_3(2,3,3) = 16 \sqrt{\frac{2}{1209}}\alpha^{-1} < 0.664$,

$\Lambda_3(3,3,3) = 3 \sqrt{\frac{7}{190}}\alpha^{-1} < 0.587$.

\item[{\bf Case $w_1 =1, w_3 \geq 2$.}]
Recall that we can fix for every $w_j>1$, $t_j = \frac{1}{2}$, so similarly as in \emph{Case $d=2,(1,2)$}, by Proposition \ref{equalKappa} we are left with a single variable, and by Proposition \ref{monotonicity_deriv}, using a binary search we can determine the location of the zeros of the derivative to arbitrary precision i.e. the interval of $t_1$ where the maximal value is attained.

\begin{tabular}{|c|c|c|}
  \hline
  Case &  Extremal $t_1$  & Upperbound Based on Left-Right Endpoint\\
  \hline
  (1,1,2) & $(0.32 , 0.33)$ & $g\left( \frac{3}{4}(1 - 0.32)^2 \right) \left( 2 T(0.33,1) + T(\frac{1}{2},2) \right)\alpha^{-1}  $\\
          &                       & $< 0.84$ \\
  \hline
  (1,1,3) & $(0.352 , 0.353)$ & $g\left( \frac{7}{8}(1 - 0.352)^2 \right) \left( 2 T(0.353,1) + T(\frac{1}{2},3) \right)\alpha^{-1}  $\\
          &                       & $< 0.7881$ \\
  \hline
  (1,2,2) & $(0.34 , 0.35)$ & $g\left( (\frac{3}{4})^2 (1 - 0.34)\right) \left( T(0.35,1) + 2 T(\frac{1}{2},2) \right)\alpha^{-1}  $\\
          &                       & $< 0.82$ \\
  \hline
  (1,2,3) & $(0.38 , 0.39)$ & $g\left( \frac{3}{4} \cdot \frac{7}{8} (1 - 0.38)  \right)\left( T(0.39,1) + T(\frac{1}{2},2)+T(\frac{1}{2},3) \right) \alpha^{-1} $\\
          &                       & $< 0.77$ \\
  \hline
  (1,3,3) & $(0.42 , 0.43)$ & $g\left( (\frac{7}{8})^2 (1 - 0.42) \right) \left( T(0.43,1) + 2T(\frac{1}{2},3) \right)\alpha^{-1}  $\\
          &                       & $< 0.72$ \\
  \hline
\end{tabular}
\end{description}

\item[{\bf Case $d=4$.}]
If $w_j > 2$, in order to maximize $\kappa_4^{\underline{w}}(\mathbf{t})$, $t_j = \frac{1}{2}$, which follows directly from
$\frac{\partial \kappa_4^{(1,1,1,3)}(\mathbf{t})}{\partial t_4} \big|_{\mathbf{t}=\mathbf{\frac{1}{2}}} = \frac{381}{497 \sqrt{994}\alpha} > 0$ and Proposition \ref{monotonicity_deriv}.

If $w_1 \geq 2$, in order to maximize $\kappa_4^{\underline{w}}(\mathbf{t})$, $\mathbf{t} = \mathbf{\frac{1}{2}}$, which follows directly from
$\frac{\partial \kappa_4^{(2,2,2,2)}(\mathbf{t})}{\partial t_4} \big|_{\mathbf{t}=\mathbf{\frac{1}{2}}} = \frac{311}{337 \sqrt{674}} > 0$ and Proposition \ref{monotonicity_deriv}.

First by Claim \ref{decomposeKappa},

$\Lambda_4(1,1,3,3) \leq \Lambda_3(1,1,3) + \Lambda_1(3) < 0.7881 + 0.2112 < 1 $,

$\Lambda_4(1,3,3,3) \leq \Lambda_3(1,3,3) + \Lambda_1(3) < 0.72 + 0.2112 < 1 $,

$\Lambda_4(1,2,3,3) \leq \Lambda_3(1,2,3) + \Lambda_1(3) < 0.77 + 0.2112 < 1 $.

\begin{description}
\item[{\bf Case $(1,1,1,1)$.}]
Similar to \emph{Case $d=2,(1,1)$}.

  By Proposition \ref{equalKappa}, since $\kappa_{4,(1,1,1,1)}'(0.24807) >0, \kappa_{4,(1,1,1,1)}'(0.24808) <0$,
  next by Proposition \ref{monotonicity_deriv}, $\Lambda_4(1,1,1,1) \leq 4 g\left( ( 1- 0.24807)^4 \right) \cdot T(0.24808,1)\alpha^{-1} < 1$.

\item[{\bf Case $(1,1,1,2)$.}]
By Proposition \ref{equalKappa}, we are left with $2$ variables, say $t_1,t_2$.
Next we show whatever value $t_2$ takes, $\Lambda_4(1,1,1,2) < 1$.
The procedure works as follows, suppose $t_2^{-} \leq t_2 \leq t_2^{+}$, then
\[\Lambda_4(1,1,1,2) < g\left( (1 - t_1)^3 (1-t_2^{-}/2) \right) \cdot \left( 3 T(t_1,1) + T(t_2^{+}, 2) \right)\alpha^{-1}.\]
Note that the right-hand-side is a single-variate function in $t_1$, so from now on it is similar to \emph{Case $d=2,(1,2)$},
 first by Proposition \ref{monotonicity_deriv}, using a binary search we can determine to arbitrary precision $t_1^{-}, t_1^{+}$ with $t_1^{-} \leq t_1 \leq t_1^{+}$ i.e. the location of zeros of the derivative, or the value of $t_1$ where the maximal value is attained. After which we simply apply a direct evaluation via
 \[\Lambda_4(1,1,1,2) < g\left( (1 - t_1^{-})^3 (1-t_2^{-}/2) \right) \cdot \left( 3 T(t_1^{+},1) + T(t_2^{+}, 2) \right)\alpha^{-1}.\]
 to get the desired upper bound.

 Denote $U(t_1^{-},t_1^{+},t_2^{-},t_2^{+}) \triangleq g\left( (1 - t_1^{-})^3 (1-\frac{t_2^{-}}{2}) \right)  \left( 3 T(t_1^{+},1) + T(t_2^{+}, 2) \right)\alpha^{-1}$.

 For instance, suppose $0 \leq t_2 < 0.2$, first via binary search we determine the extremal $t_1 \in (0.28, 0.281)$, hence $U(0.28,0.281,0,0.2) < 0.993 $.

 The following table is a case-by-case analysis. We divide the range of $t_2$, determine the range of $t_1$ where the maximal value is located, and derive an upperbound based the range of both $t_2$ and $t_1$ as $U(t_1^{-},t_1^{+},t_2^{-},t_2^{+})$.

\begin{tabular}{|c|c|c|}
  \hline
  $t_2$ & Extremal $t_1$ & $U(t_1^{-},t_1^{+},t_2^{-},t_2^{+})$\\
  \hline
  $[0,0.2)$  & $(0.28,0.281)$   & $<0.993$ \\
  \hline
  $[0.2,0.3)$  & $(0.268,0.269)$   & $<0.993$ \\
  \hline
  $[0.3,0.35)$  & $(0.26,0.264)$   &$<0.993$ \\
  \hline
  $[0.35,0.4)$  & $(0.259,0.26)$   &$<0.994$ \\
  \hline
  $[0.4,0.45)$  & $(0.255,0.256)$   &$<0.997$ \\
  \hline
  $[0.45,0.5]$  & $(0.251,0.252)$   &$< 0.9991$ \\
  \hline
\end{tabular}

\item[{\bf Case $(1,1,2,2)$.}]Similar to the above, in this case
\[U(t_1^{-},t_1^{+},t_2^{-},t_2^{+}) \triangleq g\left( (1 - t_1^{-})^2 (1-\frac{t_2^{-}}{2})^2 \right)  \left( 2 T(t_1^{+},1) + 2T(t_2^{+}, 2) \right)\alpha^{-1}.\]

\begin{tabular}{|c|c|c|}
  \hline
  $t_2$ & Extremal $t_1$ & $U(t_1^{-},t_1^{+},t_2^{-},t_2^{+})$\\
  \hline
  $[0,0.23)$  & $(0.32,0.323)$   &$< 0.9944$ \\
  \hline
  $[0.23,0.35)$  & $(0.2895,0.2896)$   &$< 0.9998$ \\
  \hline
  $[0.35,0.4)$  & $(0.27,0.28) $  &$< 0.992$ \\
  \hline
  $[0.4,0.45)$  & $(0.26,0.266)$   &$< 0.997$ \\
  \hline
  $[0.45,0.5]$  & $(0.256,0.257)$   &$< 0.997$ \\
  \hline
\end{tabular}

\item[{\bf Case $(1,1,2,3)$.}] Similar to the above except $t_4 = \frac{1}{2}$, so in this case
\[U(t_1^{-},t_1^{+},t_2^{-},t_2^{+}) \triangleq g\left( \frac{7}{8} (1 - t_1^{-})^2 (1-\frac{t_2^{-}}{2}) \right) \left( 2 T(t_1^{+},1) + T(t_2^{+}, 2)+ T(\frac{1}{2},3) \right)\alpha^{-1}.\]

\begin{tabular}{|c|c|c|}
  \hline
  $t_2$ & Extremal $t_1$ & $U(t_1^{-},t_1^{+},t_2^{-},t_2^{+})$\\
  \hline
  $[0,0.35)$  & $(0.304,0.305)$   &$< 0.996$ \\
  \hline
  $[0.35,0.5]$  & $(0.28,0.29)$   &$< 0.99$ \\
  \hline
\end{tabular}

\item[{\bf Case $(1,2,2,2)$.}]Similar to the above,

\[U(t_1^{-},t_1^{+},t_2^{-},t_2^{+}) \triangleq g\left( (1 - t_1^{-}) (1-\frac{t_2^{-}}{2})^3 \right) \left(  T(t_1^{+},1) + 3T(t_2^{+}, 2) \right)\alpha^{-1}.\]

\begin{tabular}{|c|c|c|}
  \hline
  $t_2$ & Extremal $t_1$ & $U(t_1^{-},t_1^{+},t_2^{-},t_2^{+})$\\
  \hline
  $[0,0.25)$  & $(0.39,0.4)$   &$< 0.991$ \\
  \hline
  $[0.25,0.35)$  & $(0.32,0.33)$   &$< 0.983$ \\
  \hline
  $[0.35,0.4)$  & $(0.29,0.3)$   &$< 0.98$ \\
  \hline
  $[0.4,0.45)$  & $(0.28,0.29)$   &$< 0.99$ \\
  \hline
  $[0.45,0.5]$  & $(0.26,0.27)$   &$< 0.998$ \\
  \hline
\end{tabular}

\item[{\bf Case $(1,2,2,3)$.}] Similar to the above except $t_4 = \frac{1}{2}$, so in this case

\[U(t_1^{-},t_1^{+},t_2^{-},t_2^{+}) \triangleq g\left( \frac{7}{8} (1 - t_1^{-}) (1-\frac{t_2^{-}}{2})^2 \right) \left(  T(t_1^{+},1) + 2 T(t_2^{+}, 2)+ T(\frac{1}{2},3) \right)\alpha^{-1}.\]

\begin{tabular}{|c|c|c|}
  \hline
  $t_2$ & Extremal $t_1$ & $U(t_1^{-},t_1^{+},t_2^{-},t_2^{+})$\\
  \hline
  $[0,0.3)$  & $(0.36,0.37)$   &$< 0.985$ \\
  \hline
  $[0.3,0.4]$  & $(0.32,0.33)$   &$< 0.96$ \\
  \hline
  $[0.4,0.5]$  & $(0.29,0.3)$   &$< 0.98$ \\
  \hline
\end{tabular}

\item[{\bf Case $w_1 \geq 2$.}]
In these cases we have $\mathbf{t} = \mathbf{\frac{1}{2}}$, so by direct evaluation we have

$\Lambda_4(2,2,2,2) = 2 \sqrt{\frac{2}{337}}\alpha^{-1} < 0.95$,

$\Lambda_4(2,2,2,3) = 51 \sqrt{\frac{3}{9814}}\alpha^{-1} < 0.91$,

$\Lambda_4(2,2,3,3) = 23 \sqrt{\frac{2}{1465}}\alpha^{-1} < 0.87$,

$\Lambda_4(2,3,3,3) = 41 \sqrt{\frac{7}{18462}}\alpha^{-1} < 0.82$,

$\Lambda_4(3,3,3,3) = 42 \sqrt{\frac{2}{6497}}\alpha^{-1} < 0.76$.

\item[{\bf Case $(1,1,1,3)$.}]
Recall that we can fix for every $w_j>2$, $t_j = \frac{1}{2}$, so similar to \emph{Case $d=2,(1,2)$}, by Proposition \ref{equalKappa} we are left with a single variable, and by Proposition \ref{monotonicity_deriv} we determine via a binary search that extremal $t_1 \in (0.27, 0.28)$,
\[\Lambda_4(1,1,1,3) < g\left( \frac{7}{8}(1 - 0.27)^3  \right) \cdot \left( 3 T(0.28,1) + T(\frac{1}{2},3) \right)\alpha^{-1} <0.96. \]

\end{description}

\end{description}

\qed

\section{Proof of Lemma \ref{main-lem2}}



Let $\varphi(x) \triangleq 2 \sinh ^{-1}\left(\sqrt{x}\right)$, by induction on $L$ we show
$\abs{ \varphi \circ R(C,x,L) -  \varphi \circ R(C,x)} \leq 4 \alpha^L$.
Consider the analytic version of recursion $R(C,x,L)$ as $h(\mathbf{r}) \triangleq \prod_{j=1}^{d} \frac{1}{1+\sum_i^{w_j} r_{j,i}}$,
essentially we show $\varphi \circ h \circ \varphi^{-1}$ exhibits exponential correlation decay.

Let $\mathbf{\hat{y}}$ be the estimated vector with $\hat{y}_{j,i} = \varphi \circ R\left(C_{j,i},\ x_{j,i},\ L - 1\right)$, and $\mathbf{y}$ be the true vector with $y_{j,i} = \varphi \circ R\left(C_{j,i},\ x_{j,i}\right)$. Denote $\boldsymbol{\eps} \triangleq \mathbf{\hat{y}} - \mathbf{y}$,
$\mathbf{r} \triangleq \varphi^{-1} (\mathbf{y})$.

By Mean Value Theorem, $\exists \gamma: 0\leq \gamma \leq 1, \mathbf{\tilde{y}} = \gamma \mathbf{y} + (1-\gamma) \mathbf{\hat{y}}$ such that, let $\mathbf{\tilde{r}} \triangleq \varphi^{-1} (\mathbf{\tilde{y}}) $,
\begin{equation}
\label{main-lem2-epsilon}
\varphi\circ R(C,x,L) - \varphi\circ R(C,x) =
\sum_{a,b} \frac{ \partial (\varphi \circ h \circ \varphi^{-1}) }{\partial y_{a,b}}\Big|_{\mathbf{y}=\mathbf{\tilde{y}}} \cdot \eps_{a,b} =
\sum_{a,b} \frac{\Phi(h(\mathbf{\tilde{r}}))}{\Phi(\tilde{r}_{a,b})} \left( \frac{\partial h}{\partial \tilde{r}_{a,b}}\Big|_{\mathbf{r}=\mathbf{\tilde{r}}}\right) \eps_{a,b}.
\end{equation}
Let $\eps_{max} \triangleq \max_{a,b} \set{\abs{\eps_{a,b}}}$,
\[ \abs{\varphi\circ R(C,x,L) - \varphi\circ R(C,x)}  \leq
\sum_{a,b} \frac{\Phi(h(\mathbf{\tilde{r}}))}{\Phi(\tilde{r}_{a,b})} \abs{ \left( \frac{\partial h}{\partial r_{a,b}}\Big|_{\mathbf{r}=\mathbf{\tilde{r}}} \right) } \cdot \eps_{max} =
\hat{\kappa}_d^{\underline{w}} (\mathbf{\tilde{r}})  \cdot \eps_{max}.\]

Recall that the case where $d_x(C) = 3$ is invoked at most once, by Claim \ref{single-layer-triangle} and \ref{sub-triangle}, for some $w_j, r$ we have $\hat{\kappa}_3^{\underline{w}} (\mathbf{\tilde{r}}) \leq 3 \kappa_1^{(w_j)}(r) < 2\alpha$.

As for the rest where $d_x(C) \leq 2$, we show by induction on $L$ with induction hypothesis:
\[\abs{ \varphi \circ R(C,x,L) -  \varphi \circ R(C,x)} \leq 2 \alpha^L \ \mathrm{ for } \ d_x(C) \leq 2.\]
For base case $L=0$, since $0\leq \varphi \circ R < 2$, so $\abs{ \varphi \circ R(C,x,L) -  \varphi \circ R(C,x)} < 2$.

Suppose the induction hypothesis holds for $L<l$, we prove it is true for $L=l$.

If $x$ is a free variable in $C$, i.e. $d=0$, $R(C, x, L) = R(C, x) = 1$, it is sufficient to check $d=1,2$.
And if $x$ can be inferred (due to $w_a=0$ for some $a$), $R(C,x,L) = R(C,x) = 0$.

Next by induction hypothesis, $\eps_{max} \leq 2 \alpha^{L-1}$, by Claim \ref{condition-triangle} and \ref{single-layer-triangle}, $\hat{\kappa}_d^{\underline{w}} (\mathbf{\tilde{r}})< \alpha$,
\[\abs{ \varphi \circ R(C,x,L) -  \varphi \circ R(C,x)} \leq 2 \alpha^L \ \mathrm{ for } \ d_x(C) \leq 2.\]

Hence in all we have
\[\abs{ \varphi \circ R(C,x,L) -  \varphi \circ R(C,x)} \leq 4 \alpha^L.\]

Similar to the CNF problem, this concludes the proof.
\qed

\end{document}